\DocumentMetadata{}
\documentclass[acmsmall]{acmart}

\setcopyright{cc}
\setcctype{by}
\acmJournal{PACMMOD}
\acmYear{2026} \acmVolume{4} \acmNumber{3 (SIGMOD)} \acmArticle{140}
\acmMonth{6} \acmDOI{10.1145/3802017}

\theoremstyle{definition}
\newtheorem{lemma}{Lemma}[section]
\newtheorem*{lemma*}{Lemma}

\newtheorem*{corollary*}{Corollary}

\newtheorem*{theorem*}{Theorem}
\newtheorem*{inducthyp*}{Inductive Hypothesis}
\newtheorem*{definition*}{Definition}
\newtheorem{definition}[lemma]{Definition}

\newtheorem{prob}[lemma]{Problem}
\newtheorem*{rem*}{Remark}

\usepackage{tikz,pgfplots}
\usetikzlibrary{arrows.meta}
\usetikzlibrary{shapes.geometric}
\usepackage[ruled,linesnumbered,vlined]{algorithm2e}
\usepackage{amsmath}
\usepackage{chngcntr}
\usepackage{xspace}
\usepackage{booktabs}
\usepackage{tabularx}
\usepackage{soul}
\usepackage{lipsum}
\usepackage{pifont}
\usepackage{threeparttable} 
\usepackage{placeins}
\usepackage{enumitem}
\usepackage{listings}
\usepackage[normalem]{ulem}
\usepackage{multirow}
\usepackage{svg}
\usepackage[dvipsnames, svgnames]{xcolor}
\usepackage{tablefootnote}
\usepackage{filecontents}
\usepackage{fontawesome}
\usepackage{tcolorbox}
\usepackage{footmisc}
\usepackage[noabbrev,capitalize,nameinlink]{cleveref}
\usepackage{dashbox}
\tcbuselibrary{listings,breakable}

\definecolor{darkgreen}{rgb}{0.0, 0.5, 0.0}
\definecolor{darkred}{rgb}{0.82, 0.1, 0.26}
\definecolor{mycolor}{rgb}{0.1, 0.5, 0.8}

\newcommand{\placeholder}[1]{%
  {\texttt{\textcolor{NavyBlue}{#1}}}%
}

\newcommand{\modify}[1]{{\color{black} {#1}}}
\newcommand{\modifya}[1]{{\color{black} {#1}}}
\newcommand{\modifyb}[1]{{\color{black} {#1}}}
\newcommand{\modifyc}[1]{{\color{black} {#1}}}

\newcommand{\toolname}{\textsf{Argus}\xspace}
\newcommand{\bugsnumber}{41\xspace}
\newcommand{\confirmednumber}{36\xspace}
\newcommand{\fixednumber}{27\xspace}
\newcommand{\logicnumber}{36\xspace}
\newcommand{\othernumber}{5\xspace}
\newcommand{\sqlkeyword}[1]{{\ttfamily {#1}}}

\definecolor{NavyBlue}{rgb}{0,0,0.5}
\definecolor{lstString}{HTML}{576574}
\definecolor{myGreen}{HTML}{009432}
\definecolor{bugColor}{HTML}{c0392b}
\definecolor{sqlinlineColor}{HTML}{000000}

\lstdefinestyle{sqlstyle}{
    language=SQL,
    basicstyle=\ttfamily,    
    keywordstyle=\textbf,
    keywordstyle=[2]{\textbf},
    morekeywords={WITH, TEXT, BOOLEAN, LATERAL, MOD, ANALYZE, EXIST, ASOF, IS, EXPLAIN, STRICT, RETURN, json_valid, json_array_length, json_array},
    morekeywords=[2]{query, clause, pattern, ret},
    stringstyle=\color{lstString},
    commentstyle=\color{gray},
    backgroundcolor=\color{white!5},
    frame=single,
    framerule=0pt,
    rulecolor=\color{black},
    numbers=left,
    numberstyle=\footnotesize,
    numbersep=7pt,
    showspaces=false,
    showstringspaces=false,
    keepspaces=true,        
    showtabs=false,
    tabsize=4,
    flexiblecolumns=true,
    breaklines=true,
    breakatwhitespace=false,
    breakautoindent=true,
    breakindent=1em,
    escapeinside=``,
    xleftmargin=2em,  
    framexleftmargin=4ex,
    captionpos=t,
}

\newif\ifcomments
\commentstrue
\ifcomments
    \providecommand{\alvin}[1]{{\color{brown}{alvin: #1 }}}
    
\else
    \providecommand{\alvin}[1]{}
\fi

\Crefname{section}{Sec.}{Sec.}
\Crefname{algorithm}{Alg.}{Alg.}
\Crefname{figure}{Fig.}{Fig.}

\author{Qiuyang Mang}
\authornote{Corresponding authors: Qiuyang Mang (qmang@berkeley.edu) and Huanchen Zhang (huanchen@tsinghua.edu.cn).}
\affiliation{%
  \institution{UC Berkeley}
  \country{USA}
}
\email{qmang@berkeley.edu}

\author{Runyuan He}
\affiliation{%
  \institution{UC Berkeley}
  \country{USA}
}

\author{Suyang Zhong}
\affiliation{%
  \institution{National University of Singapore}
  \country{Singapore}
}

\author{Xiaoxuan Liu}
\affiliation{%
  \institution{UC Berkeley}
  \country{USA}
}

\author{Huanchen Zhang}
\affiliation{%
  \institution{Tsinghua University}
  \country{China}
}
\email{huanchen@tsinghua.edu.cn}

\author{Alvin Cheung}
\affiliation{%
  \institution{UC Berkeley}
  \country{USA}
}

\renewcommand{\shortauthors}{Qiuyang Mang et al.}



\begin{document}

\title{Automated Discovery of Test Oracles for Database Management Systems Using LLMs}

\received{October 2025}
\received[revised]{January 2026}
\received[accepted]{February 2026}

\begin{abstract}
Since 2020, automated testing for Database Management Systems (DBMSs) has flourished, uncovering hundreds of bugs in widely-used systems. A cornerstone of these techniques is \emph{test oracle}, which typically implements a mechanism to generate equivalent query pairs, and subsequently runs the pair and identifies bugs by checking the consistency of their results. While running these oracles can be automated, designing the mechanism to generate equivalent queries remains a fundamentally manual endeavor. This paper explores the use of large language models (LLMs) to automate the discovery of equivalent queries in the design of test oracles, addressing a long-standing bottleneck towards fully automated DBMS testing. 

Although LLMs demonstrate impressive creativity, they are prone to hallucinations that can produce numerous false positive bug reports. Furthermore, their high monetary cost and latency mean that LLM invocations should be limited to ensure that bug detection is efficient and economical.
To this end, we introduce \toolname, a novel framework built upon the core concept of the \emph{Constrained Abstract Query}---a SQL skeleton containing placeholders and their associated instantiation conditions, \emph{e.g.,} the placeholder must be filled by a Boolean column. \toolname uses LLMs to generate pairs of these skeletons, with their equivalence formally proven using a SQL equivalence solver to ensure soundness. After that, the placeholders in the verified skeletons are instantiated with concrete, reusable SQL snippets that are also synthesized by LLMs to produce complex test cases.
We have implemented \toolname and evaluated it on five extensively tested DBMSs, discovering \bugsnumber previously unknown bugs, \logicnumber of which are logic bugs, with \confirmednumber confirmed and \fixednumber already fixed by the developers. \modifyc{The artifacts for \toolname are available at \url{https://github.com/joyemang33/Argus}}
\end{abstract}
\maketitle

\begin{section}{Introduction}
    \label{sec:introduction}

    Database Management Systems (DBMSs) are a foundational component of modern software, yet their complexity makes them prone to bugs that can compromise application behavior and data integrity.
    Logic bugs are particularly insidious; they cause a DBMS to return incorrect results without raising errors, therefore silently corrupting downstream applications~\cite{rigger2020detecting}.
    In response, the research community has developed automated testing techniques~\cite{rigger2020testing,rigger2020detecting,rigger2020finding,jiang2024detecting,zhang2025constant} that have discovered hundreds of bugs in real-world systems.

    The most critical component of these techniques is the \emph{test oracle}~\cite{barr2014oracle} that can determine the correctness of a query's output without access to the ground truth.
    Given a query, oracles for DBMS testing implement a transformation mechanism to generate a semantically equivalent variant.
    Then, by executing both queries and checking for result consistency, these oracles can detect logic bugs.
    For example, Ternary Logic Partitioning (TLP)~\cite{rigger2020finding} is a highly effective oracle that partitions a query $Q$ based on a predicate $P$ and then checks equivalence between $Q$ and the union of its three-way partition
    $Q$ \sqlkeyword{WHERE} $P$, $Q$ \sqlkeyword{WHERE} \sqlkeyword{NOT} $P$, and $Q$ \sqlkeyword{WHERE} $P$ \sqlkeyword{IS} \sqlkeyword{NULL}. 
    If a DBMS returns different results for the original query and its partitioned version, that indicates a bug.

    While human-designed oracles have found many bugs (reported in over 20 top-conference papers), the manual creation of oracles is a bottleneck, trapping researchers in a cycle of designing increasingly specialized oracles to find bugs missed by previous ones.
    For example, a TiDB~\cite{huang2020tidb} bug introduced in 2019~\cref{lst:tidb-except-bug} went undetected for years despite extensive testing\footnote{\href{https://github.com/pingcap/tidb/blob/7de620055dd290fde1d90368688b40895bec506b/types/mydecimal.go\#L1279}{Commit 7de6200}, introduced in 2019.}, because it required a specific oracle to check that any query $Q$ \sqlkeyword{EXCEPT} $Q$ should yield an empty result.
    This example highlights a fundamental challenge: manually designed oracles are not only difficult to conceive but also tend to overlook bugs. 
    Automating the discovery of test oracles is therefore essential to enable scalable DBMS testing.

\begin{figure}[t]
\vspace*{4ex}
    \centering
\begin{lstlisting}[label={lst:tidb-except-bug}, caption={Motivating bug: incorrect {\ttfamily EXCEPT} result in TiDB, which can only be detected by a specific test oracle.}, captionpos=b, escapeinside={(*}{*)},]
CREATE TABLE t1(c INT);
INSERT INTO t1 VALUES (1);
SELECT c / 3 FROM t1 WHERE false; -- {} (*\faCheck{}*)
SELECT c / 3 FROM t1 EXCEPT SELECT c / 3 FROM t1;
-- {0.3333} (*\faBug{}*) 
\end{lstlisting}
\vspace*{-4ex}
\end{figure}

    Large Language Models (LLMs), with their success in code generation~\cite{hong2025autocomp, novikov2025alphaevolve}, offer a promising avenue for automating this process.
    A naive approach would be to prompt an LLM to generate a semantically equivalent variant for a given seed SQL query.
    However, this strategy suffers from two limitations:

    \begin{enumerate}[leftmargin=*, itemsep=2pt]
        \item \textbf{Scalability:} LLM invocations are costly and have high latency.
        Compared to traditional SQL generators~\cite{Seltenreich2022sqlsmith, fu2023griffin}, this is infeasible for modern DBMS testing, which often requires executing thousands of queries per minute to find bugs efficiently, such as in SQLancer~\cite{ba2023testing}.

        \item \textbf{Soundness:} LLMs are prone to hallucination and may generate query pairs that are not truly equivalent.
        Such unsound oracles produce false positives, thereby undermining the reliability of the testing process.
        Note that verifying equivalence empirically by running the query pairs on multiple database instances is unreliable, as this cannot guarantee the complete removal of semantically inequivalent pairs~\cite{he2024verieql}, and may also filter out true bugs that might exist on all databases under test~\cite{slutz1998massive}.
    \end{enumerate}

    \noindent\paragraph{Key insights}
    To address these challenges, we propose \toolname, a novel, fully automated framework for finding logic bugs in DBMSs.
    \toolname uses a two-stage process that separates oracle discovery from test case generation. First, it leverages an LLM in an offline phase to discover reusable test oracles. Then, these oracles are formally verified for correctness before being instantiated into thousands of concrete test cases.
    Since generating abstract test oracles is much cheaper and faster than directly using LLMs to craft concrete test cases, this approach tackles both scalability and soundness head-on. Intuitively, these abstract oracles are parameterized query templates that can be instantiated many times with different SQL snippets. 

    
    \vspace{-1ex}
    \paragraph{Test oracle representation}
    \toolname introduces a novel and expressive representation for test oracles that can be easily understood by LLMs called the \emph{equivalent Constrained Abstract Query (CAQ) pair}. 
    A CAQ is an abstract SQL query with {\em placeholders} that can represent tables, columns, predicates, and other SQL components.
    These placeholders are also associated with the corresponding constraints, which are essential to implement test oracles.
    For example, a CAQ with two placeholders might specify that placeholders must be syntactically identical, or that a predicate must be a Boolean expression when instantiated.
    An {\em equivalent CAQ pair} consists of two CAQs that are semantically identical for every possible instantiation of placeholders, 
    \emph{i.e.,} for all possible fillings of the placeholders with concrete SQL snippets that satisfy their constraints.
    CAQs allow \toolname to represent a wide range of test oracles, including several existing ones.
    For instance, TLP~\cite{rigger2020finding} can be represented as an equivalent CAQ pair, \emph{i.e.,} $Q_1$ and $Q_2$, as shown in Listing~\ref{lst:tlp-in-caq}. 
    In this representation, $Q_1$ serves as the original query template, while $Q_2$ is its three-way partitioned version. 
    Here, to ensure the syntax validity of the instantiated queries from the two CAQs, $\square_1$ is restricted to represent any table, and $\square_2$ must be a Boolean expression constructed using the columns of table $t_1$.
    These two CAQs are semantically equivalent for all valid instantiations of the placeholders that satisfy their constraints, where $\square_1$ and $\square_2$ should also be instantiated consistently in both queries.

\begin{figure}[t]
\vspace*{4ex}
    \centering
\begin{lstlisting}[label={lst:tlp-in-caq}, caption={An example of representing and instantiating TLP~\cite{rigger2020finding} oracle in CAQ.}, captionpos=b, escapeinside={(*}{*)},]
CREATE TABLE t1(c0 VARCHAR, ...);
CREATE TABLE t2(...);
SELECT * FROM t1, (*${\square_1}\triangleright\placeholder{\textsf{Table}(...)}$*);          -- (*\textcolor{gray}{$Q_1$}*) 
SELECT * FROM t1, (*${\square_1}\triangleright\placeholder{\textsf{Table}(...)}$*) 
WHERE ((*${\square_2}\triangleright\placeholder{\textsf{Expr}(t1:BOOLEAN)}$*) IS TRUE) UNION ALL 
SELECT * FROM t1, (*${\square_1}\triangleright\placeholder{\textsf{Table}(...)}$*) 
WHERE ((*${\square_2}\triangleright\placeholder{\textsf{Expr}(t1:BOOLEAN)}$*) IS FALSE) UNION ALL 
SELECT * FROM t1, (*${\square_1}\triangleright\placeholder{\textsf{Table}(...)}$*) 
WHERE ((*${\square_2}\triangleright\placeholder{\textsf{Expr}(t1:BOOLEAN)}$*) IS NULL);     -- (*\textcolor{gray}{$Q_2$}*) 
(*${\square_1}\triangleright\placeholder{\textsf{Table}(...)} \mapsto$*) (*\texttt{\textcolor{NavyBlue}{t1 ASOF JOIN t2}}*)
(*${\square_2}\triangleright\placeholder{\textsf{Expr}(t1:BOOLEAN)} \mapsto$*) (*\texttt{\textcolor{NavyBlue}{json\_valid(t1.c0)}}*)\end{lstlisting}
\vspace*{-4ex}
\end{figure}




    \vspace{-1ex}
    \paragraph{Test oracle generation, verification, and instantiation}
    To generate candidate oracles, \toolname{} prompts an LLM with an initial CAQ 
    and instructs it to produce a semantically equivalent one. 
    To address the \emph{soundness challenge}, each generated CAQ pair is then formally verified. 
    We employ SQLSolver~\cite{ding2023proving}, a state-of-the-art SQL equivalence prover, to confirm their equivalence. 
    If the prover successfully validates the equivalence, the CAQ pair is accepted as a valid test oracle; otherwise, it is discarded. 
    Since all existing provers are designed for concrete queries, they cannot directly handle CAQs. We bridge this gap by {\em instantiating the placeholders with virtual tables and columns}, thus transforming them into concrete queries.
    This verification step ensures that only provably correct oracles are used for testing.
    
    To tackle the \emph{scalability challenge}, each verified oracle is instantiated into thousands of concrete test cases. 
    Specifically, we populate the placeholders in each CAQ pair from a large corpus of pre-generated SQL snippets.
    This approach allows us to generate test cases for novel and complex features provided by the DBMS under test, even though they are not supported by our prover.
    For instance, in Listing~\ref{lst:tlp-in-caq}, the placeholders $\square_1$ and $\square_2$ can be instantiated with various table joins and Boolean expressions, respectively, such as \texttt{t1} \texttt{ASOF} \texttt{JOIN} \texttt{t2} and \texttt{json\_valid(t1.c0)}.
    These snippets can be generated offline by a hybrid approach combining an LLM to cover diverse database features and a high-throughput generator, such as SQLancer~\cite{rigger2020testing}, even though the prover currently cannot reason about them. 

    \vspace{-1ex}
    \paragraph{Evaluation}
    We have implemented \toolname and evaluated it on five widely used and comprehensively tested DBMSs: Dolt~\cite{dolt2025}, DuckDB~\cite{raasveldt2019duckdb}, MySQL~\cite{mysql2025}, PostgreSQL~\cite{momjian2001postgresql}, and TiDB~\cite{huang2020tidb}.
    \toolname discovered \bugsnumber unique and previously unknown bugs. 
    Of these, \confirmednumber have been confirmed and \fixednumber have been fixed by the developers. The bugs discovered include \logicnumber critical logic bugs that lead to incorrect query results, while the remaining \othernumber cause performance problems or crashes.
    In our empirical comparisons with state-of-the-art open source testing tools in DuckDB~\cite{raasveldt2019duckdb}, \toolname demonstrated an improvement of up to 1.19$\times$ in code coverage, and 6.43 $\times$ in metamorphic coverage~\cite{ba2025metamorphic}, a recently developed coverage indicator to assess the ability to find logic bugs.
    In addition, we also compared \toolname with a union of existing test oracles, rewriting them into equivalent CAQ pairs equipped with the same query generator.
    The results show that \toolname's new oracles can detect 3.33$\times$ more unique logic bugs than the sum of prior oracles from these works~\cite{rigger2020detecting,rigger2020finding,ba2024keep,jiang2024detecting} 
    within 6 hours of testing on Dolt~\cite{dolt2025}.
    Our ablation studies further validate the effectiveness of the SQL equivalence prover in ensuring soundness and the contribution of CAQ to enhancing scalability.
    \begin{figure}[tp!]
        \centering
        \includegraphics[width=1.0\linewidth]{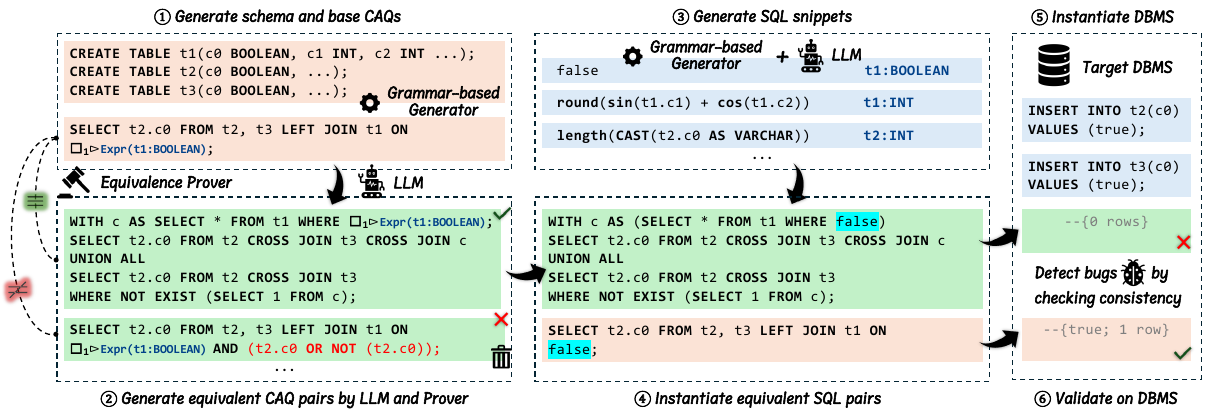}
        \caption{Overall pipeline of \toolname.}
        \label{fig:pipeline}
        \vspace{-2ex}
    \end{figure}
    This paper makes the following contributions.
    \begin{itemize}[leftmargin=*, itemsep=2pt]
        \item We introduce the concept of the \emph{equivalent CAQ pair} --- a novel and expressive representation for test oracles in DBMSs --- that can also capture features not supported by existing SQL equivalence provers. 
        Based on that, we propose a new methodology that leverages LLMs to automatically discover test oracles for DBMSs by using them to generate CAQs. To our knowledge, we are the first tool that utilizes LLMs to generate DBMS test oracles.

        \item We address the \emph{soundness challenge} by formally verifying the correctness of CAQ pairs, and address the \emph{scalability challenge} by pre-generating a large corpus of reusable SQL snippets, which can then be used to efficiently instantiate equivalent CAQ pairs into numerous concrete tests. 

        \item We perform an \emph{extensive evaluation} on five widely tested DBMSs, uncovering \bugsnumber previously unknown bugs, which outperforms the state-of-the-art techniques in several coverage metrics and the number of unique logic bugs found.
    \end{itemize}
\end{section}

\begin{section}{Background}

    \paragraph{Test oracle}
    A test oracle is a mechanism for verifying whether a system’s output is correct for a given input~\cite{barr2014oracle}. In the context of DBMS testing, test oracles typically operate by transforming a given SQL query into a semantically equivalent variant, as in TLP~\cite{rigger2020finding}, NoREC~\cite{rigger2020detecting}, and EET~\cite{jiang2024detecting}.
    By doing so, the oracle can compare the execution results of the original and transformed queries to detect potential logic bugs.
    In this work, CAQ provides a unified framework for LLMs to generate test oracles automatically.
    We show that most of these prior oracles can be formalized as instances of equivalent CAQ pairs, such as that shown in Listing~\ref{lst:tlp-in-caq}.
    We include more examples in Appendix~C.

    \vspace{-1.5ex}
    \paragraph{SQL equivalence verification}
    The goal of SQL equivalence verification is to determine if two SQL queries are semantically equivalent, meaning they produce the same result when executed on all possible input database instances.
    SQL equivalence is generally undecidable~\cite{trakhtenbrot1950impossibility}, but partial deciders exist and are widely used in various scenarios such as query rewriting~\cite{wang2022wetune,liu2022leveraging} and text-to-SQL~\cite{yang2025automated}.
    One line of work has focused on SQL equivalence provers that are sound (\emph{i.e.,} they only confirm true equivalences) but not necessarily complete (\emph{i.e.,} they may fail to identify all equivalent pairs)~\cite{wang2024qed,chu2017cosette,ding2023proving, chu2017hottsql}.
    On the other hand, there are also some works on SQL equivalence disprovers~\cite{zhao2025polygon,he2024verieql}, which can conservatively produce counterexamples for nonequivalent query pairs.

    \vspace{1ex}
    In this paper, we use LLMs to generate test oracles and use SQL provers to validate their equivalence. However, existing provers are limited, as they only support concrete SQL pairs and can reason about limited SQL features. To bridge this gap, we extend the provers to handle CAQs by introducing virtual tables and columns. 
    By doing so, our framework can support many more features by instantiating these placeholders in the verified test oracles.
\end{section}

\begin{section}{\toolname Overview}

    We now use the detection of a previously unknown logic bug in DuckDB~\cite{raasveldt2019duckdb} as a case study to provide an overview of the \toolname framework.

        \Cref{fig:pipeline} presents the overall pipeline of \toolname, which consists of two stages: \emph{Test Oracle Discovery} (\ding{172} and \ding{173}) and \emph{Test Cases Instantiation} (\ding{174} -- \ding{177}). The first stage aims to automatically discover a large number of high-quality test oracles in the form of equivalent CAQ pairs, while the second stage focuses on deriving concrete test cases from these CAQ pairs to detect bugs in DBMSs.

        In \ding{172}, we generate a set of CAQs with their associated schemas by a grammar-based generator without using LLMs, and this will be detailed in Section~\ref{sec:database-seeding}.
        Specifically, we use SQLancer++~\cite{zhong2025scaling}'s query generator to produce seed queries, but we do not use their predefined test oracles (such as TLP~\cite{rigger2020finding} and NoREC~\cite{rigger2020detecting}) for bug detection.
        For example, we generate the following CAQ:
         \begin{lstlisting}[frame=none,numbers=none, belowskip=-1.5ex, xleftmargin=\parindent,escapeinside={(*}{*)}]
SELECT t2.c0 FROM t2, t3 LEFT JOIN t1 
ON (*$\square_1 \triangleright\placeholder{\textsf{Expr}(t1:BOOLEAN)}$*);
        \end{lstlisting}
        After that, in \ding{173} these CAQs will serve as seed queries for the subsequent LLM-based test oracle discovery.
        First, we employ an iterative prompting strategy to guide LLMs to generate various equivalent CAQs for each seed query, which will be elaborated in Section~\ref{sec:caq-pairs-generation}.
        For instance, given the above CAQ as input, the LLM may generate the following CAQs:
        \begin{lstlisting}[frame=none,numbers=none, belowskip=-1.5ex, xleftmargin=\parindent,escapeinside={(*}{*)}]
(*$C_1$*): WITH c AS (SELECT * FROM t1 WHERE (*$\square_1 \triangleright\placeholder{\textsf{Expr}(t1:BOOLEAN)}$*))
SELECT t2.c0 FROM t2 CROSS JOIN t3 CROSS JOIN c
UNION ALL
SELECT t2.c0 FROM t2 CROSS JOIN t3
WHERE NOT EXIST (SELECT 1 FROM c);

(*$C_2$*): SELECT t2.c0 FROM t2, t3 LEFT JOIN t1 
ON (*$\square_1 \triangleright\placeholder{\textsf{Expr}(t1:BOOLEAN)}$*) AND (t2.c0 OR NOT (t2.c0));
        \end{lstlisting}
        Recall that LLM-generated CAQs are not guaranteed to be semantically equivalent due to the hallucinations. 
        The first candidate, \emph{i.e.,} $C_1$, is indeed a valid transformation. 
        Its equivalence holds because it correctly expands the \sqlkeyword{LEFT JOIN} clause into a \sqlkeyword{UNION ALL}, and then uses a Common Table Expression (CTE) to reposition the \sqlkeyword{WHERE} clause.
        However, the second candidate, \emph{i.e.,} $C_2$, is invalid since 
        as it fails to consider the corner case of \texttt{NULL} values in SQL's three-valued logic.

        To filter out such inequivalent candidates, we employ a SQL equivalence prover called SQLSolver~\cite{ding2023proving} for this purpose. 
        Specifically, we will discard any candidate that cannot be proven equivalent to the input query. The details about how to communicate with CAQ pairs and concrete SQL provers will be presented in Section~\ref{sec:equivalence-checking}.

        In the next stage, our task is to instantiate each verified CAQ pair into multiple concrete SQL query pairs by filling the placeholders in the CAQs.
        Before that, we generate a corpus of SQL snippets for placeholder filling, ensuring efficiency and cost-effectiveness during testing through reuse.
        This corpus also needs to be scalable, complex, and diverse enough to cover various database features. 
        
        In \ding{174}, we use a hybrid approach that integrates LLMs with a grammar-based generator, producing complex, feature-rich snippets (by LLMs) while maintaining scalability (by grammar-based generator). 
        Note that LLMs may not generate the snippets that return the specified datatypes, or even generate invalid SQL snippets that will cause false positives during testing, such as non-deterministic features.
        To address this, we employ runtime validation to filter snippets.
        This will be discussed in Section~\ref{sec:corpus-generation}.

        With the generated corpus, in \ding{175} we  instantiate each CAQ pair into multiple concrete SQL query pairs by filling the placeholders with snippets from the corpus.
        The details about how to reuse the corpus for instantiating multiple CAQ pairs associated with different schemas will be presented in Section~\ref{sec:query-instantiation}. 
        For example, for the above equivalent CAQ pair, we must fill the placeholder $\square_1 \triangleright\placeholder{\textsf{Expr}(t1:BOOLEAN)}$ with a snippet that produces a Boolean value from the table \texttt{t1}.
        For instance, the generated snippet \texttt{false} can be used to instantiate the CAQ pair into the following concrete SQL query pair:
        \begin{lstlisting}[frame=none,numbers=none, belowskip=-1.5ex, aboveskip=1ex,xleftmargin=\parindent,escapeinside={(*}{*)}]
(*$Q_1$:*) SELECT t2.c0 FROM t2, t3 LEFT JOIN t1 ON false;
(*$Q_2$:*) WITH c AS (SELECT * FROM t1 WHERE false) ...
        \end{lstlisting}
        
        Once we have the concrete equivalent SQL query pairs derived by LLM-generated test oracles, \emph{e.g.,} $Q_1 \equiv Q_2$ in the above example, in steps \ding{176} and \ding{177} we can execute them on the target DBMS to check for bugs.
        For example, when we insert two rows \texttt{t2(c0:true)} and \texttt{t3(c0:true)} into an empty database, respectively, both queries should return one row. 
        However, in DuckDB v1.4.0, $Q_1$ returns one row while $Q_2$ returns zero rows, indicating a logic bug in DuckDB.
        The root cause is that DuckDB incorrectly assumes that an empty materialized CTE will always cause the outer query to return no rows,
        which does not hold when there is a \sqlkeyword{UNION ALL} operator.
        The details of database instantiation and bug reporting will be discussed in \Cref{sec:database-instantiation}.

\end{section}

\begin{section}{Constrained Abstract Query}

        In this section, we formally introduce \emph{Constrained Abstract Query (CAQ)}, the core concept of \toolname.

        A CAQ (Constrained Abstract Query) is a parameterized SQL query template that contains placeholders to be instantiated with concrete SQL fragments, together with constraints specifying valid instantiations for each placeholder.
        \Cref{fig:bnf-grammar} formalizes the structure of CAQ pairs using a BNF grammar.
        A CAQ pair consists of two CAQs and a shared schema that defines all referenced tables.
        Each CAQ is represented as a tuple $(\mathsf{Query}, \mathsf{PlaceholderMap})$, where 
        $\mathsf{Query}$ is a \texttt{SQLSelectStmt} containing placeholders, 
        and $\mathsf{PlaceholderMap}$ is a set that defines a constraint for each placeholder.
        Each placeholder $\square_i$ is constrained to be either an $\mathsf{Expr}$ or a $\mathsf{Table}$.
        An $\mathsf{Expr}$ placeholder denotes an expression whose scope is bound to a specific source table and whose return type is specified (e.g., a \texttt{BOOLEAN} column reference, an expression that returns an \texttt{INT} or \texttt{TEXT}).
        A $\mathsf{Table}$ placeholder, on the other hand, denotes a table or a subquery, and its constraint specifies the expected schema (column names and datatypes) of the subquery result.
        $\mathsf{Table}$ placeholders may appear wherever an SQL table expression is valid, 
        such as in a \sqlkeyword{FROM} clause or in nested subqueries,
        while $\mathsf{Expr}$ placeholders may appear wherever an expression can be evaluated 
        within the scope of their associated table, such as in the \texttt{SELECT} clause or filter conditions.
        The two types of placeholders are designed to capture complementary SQL features:
        $\mathsf{Expr}$ placeholders model expressions involving functions and operators,
        whereas $\mathsf{Table}$ placeholders model compositional constructs involving subqueries, 
        \sqlkeyword{JOIN}, \sqlkeyword{EXISTS}, and \sqlkeyword{IN}.

        We now define how we instantiate a CAQ and how to check equivalence between a CAQ pair. 

        \begin{figure}[t]
\centering

\[
\begin{aligned}
\mathsf{CAQ\ pair}  ::=&\  \mathsf{Schema},\ \mathsf{CAQ},\ \mathsf{CAQ} \\
\mathsf{Schema} ::=&\  \texttt{SQLTableDef} \mid (\texttt{SQLTableDef},\ \mathsf{Schema}) \\
\mathsf{CAQ}  ::=&\  (\mathsf{Query},\ \mathsf{PlaceholderMap}) \\
\mathsf{Query} ::=&\  \texttt{SQLSelectStmt}[\square_1, \square_2, \ldots \ ] \\
\mathsf{PlaceholderMap} ::=& \{ \square_1\ \triangleright \mathsf{Constraint}, \square_2\ \triangleright \mathsf{Constraint}, \ldots\ \} \\
\mathsf{Constraint} ::=&\ \textcolor{NavyBlue}{\mathsf{Expr}(\mathsf{TableName}:\texttt{SQLDatatype})} \mid \\
&\ \textcolor{NavyBlue}{\mathsf{Table}(\texttt{SQLTableDef})}
\end{aligned}
\]



	\vspace{-0.2in}
    \caption{BNF Grammar for CAQ Pairs, where we omit the full SQL grammar for simplicity.}
	\label{fig:bnf-grammar}
    \vspace{-.1in}
\end{figure}

	\begin{definition}[CAQ Instantiation]
		Given a $\mathsf{CAQ}$ $q$ and its associated $\mathsf{PlaceholderMap}$ $\{\square_1 \triangleright c_1, \square_2 \triangleright c_2, \ldots, \square_k \triangleright c_k\}$, an \emph{instantiation} of $q$, denoted $q^*$, is obtained by replacing each placeholder $\square_i$ in $q$ with a concrete SQL snippet $s_i$ that satisfies the corresponding constraint $c_i$ and a set of general constraints $\mathcal{C}$. The substitution can be written as:
		\[
			q^{*} = q[\square_1 \mapsto s_1, \square_2 \mapsto s_2, \ldots, \square_k \mapsto s_k]
		\]
	\end{definition}
        
        Specifically, each $c_i$ refers to the \textsf{Constraint}s shown in \Cref{fig:bnf-grammar}.
        In addition, to prevent false positive bug reports when using the instantiated query $q^*$ for testing, we must also satisfy a set of general constraints $\mathcal{C}$, such as avoiding non-deterministic features, \emph{e.g.,} \sqlkeyword{RANDOM} and \sqlkeyword{CURRENT\_TIMESTAMP}, and ensuring the snippet is valid SQL.
        Note that these are also widely adopted constraints in prior work on manually designed test oracles~\cite{rigger2020finding,jiang2024detecting,zhang2025constant,ba2023testing}. 
        The details of $\mathcal{C}$ will be discussed in \Cref{sec:query-instantiation}.

        We say that two CAQs are equivalent if all their possible instantiations are semantically equivalent.

        \begin{definition}[Equivalent CAQ Pair]
                Given two $\mathsf{CAQ}$s $q_1$ and $q_2$ defined over the same $\mathsf{Schema}$ $s$ (which specifies the tables referenced in both queries) with the same $\mathsf{PlaceholderMap}$ $\{\square_1 \triangleright c_1, \square_2 \triangleright c_2, \ldots, \square_k \triangleright c_k\}$, we say that $q_1$ and $q_2$ form an \emph{Equivalent CAQ Pair}, denoted $(s, q_1, q_2)$, if for every possible instantiated concrete queries (Definition 4.1), $q_1^*$ and $q_2^*$, we have $q_1^* \equiv q_2^*$, \emph{i.e.,} they will return the same result set when executed on any database instance conforming to schema $s$.
        \end{definition}

        According to the above definition, once we verify that two CAQs form an equivalent pair, we can derive multiple concrete equivalent SQL query pairs by instantiating the CAQs with various snippets that satisfy the constraints of the placeholders.
        This allows us to decompose the test oracle discovery problem to CAQ pair generation and equivalence checking problems. 

\end{section}
\begin{section}{Test Oracle Discovery}
\label{sec:test-oracle-discovery}

   \begin{algorithm}[t]
        \caption{\small Discovering New Test Oracles}\label{alg:test-oracle-discovery}
        \small
        \KwIn{
            A DBMS $\mathcal{D}$, an LLM $\mathcal{M}$, a prover $\mathcal{P}$, a grammar-based generator $\mathcal{G}$, the number of schemas $N$, the number of maximum trials for each seed query $\mathsf{MaxDepth}$.
        }
        \KwOut{A set of test oracles $\mathcal{O}$.}
        \SetKwFunction{FMain}{$\mathsf{GenerateEqual}$}
        \SetKwProg{Fn}{Function}{:}{}
        \Fn{\FMain{$\mathcal{M}, \mathcal{D}, \mathcal{P}, q$, $\mathsf{StopThreshold} \leftarrow 3$}}{\label{line:func-generate-equal}
            $\mathsf{Equal} \leftarrow \{\}$, $\mathsf{Fail} \leftarrow \{\}$, $\mathsf{FailCount} \leftarrow 0$, $\mathsf{Count} \leftarrow 0$; \\
            \While {$\mathsf{FailCount} < \mathsf{StopThreshold}$} { \label{line:while-loop}
                $q' \leftarrow \mathcal{M}(q, \mathsf{Sample(Equal)}, \mathsf{Sample(Fail)})$; \label{line:prompt-llm}\\
                \If{$\mathcal{P} \vdash (q \equiv q') \texttt{ and } \mathcal{D} \models q'$ \tcp{\Cref{sec:equivalence-checking}} \label{line:equivalence-check}}{ 
                    $\mathsf{Equal} \leftarrow \mathsf{Equal} \cup \{q'\}$; \label{line:add-to-equal}\\
                    $\mathsf{FailCount} \leftarrow 0$; \\
                    $\mathsf{StopThreshold} \leftarrow \mathsf{StopThreshold} + 1$; \label{line:increase-threshold}
                }
                \Else{
                    $\mathsf{Fail} \leftarrow \mathsf{Fail} \cup \{q'\}$; \label{line:add-to-fail} \\
                    $\mathsf{FailCount} \leftarrow \mathsf{FailCount} + 1$; \label{line:increase-fail-count} \\
                }
                $\mathsf{Count} \leftarrow \mathsf{Count} + 1$; \\
                \If{$\mathsf{Count} \geq \mathsf{MaxDepth}$ \label{line:check-max-depth}}{
                    \textbf{break};
                }
            }
            \Return $\mathsf{Equal}$; \label{line:return-oracles}
        }

        $\mathcal{O} \leftarrow \emptyset$\;
        \For{$i \leftarrow 1$ \KwTo $N$ \label{line:for-loop}}{
            $s \leftarrow \mathsf{GenerateSchema}(\mathcal{G}, \mathcal{D})$ \tcp*{\Cref{sec:database-seeding}} \label{line:generate-schema}
            $q \leftarrow \mathsf{GenerateCAQ}(\mathcal{G}, \mathcal{D}, s)$ \tcp*{\Cref{sec:database-seeding}} \label{line:generate-caq}
            $\{q_1, q_2, \ldots, q_k\} \leftarrow \mathsf{GenerateEqual}(\mathcal{M}, \mathcal{D}, \mathcal{P},q)$; \label{line:generate-equal}
            \\
            $\mathcal{O} \leftarrow \mathcal{O} \cup \{(s, q, q_1), (s, q, q_2), \ldots, (s, q, q_k)\}$;
        }
        \Return $\mathcal{O}$; 
    \end{algorithm}
    We now present our test oracle discovery algorithm, outlined in Algorithm~\ref{alg:test-oracle-discovery} and consists of three phases: database seeding, CAQ pairs generation, and equivalence checking. It leverages four key components as input: a target DBMS $\mathcal{D}$ (\emph{e.g.,} DuckDB~\cite{raasveldt2019duckdb}), an LLM $\mathcal{M}$ (\emph{e.g.,} GPT o4-mini~\cite{openai-o4-mini}), a SQL prover $\mathcal{P}$ (\emph{e.g.,} SQLSolver~\cite{ding2023proving}), and a grammar-based generator $\mathcal{G}$ (\emph{e.g.,} SQLancer~\cite{ba2023testing}).

\setstcolor{black}

    \begin{subsection}{Database Seeding}
        \label{sec:database-seeding}
        To discover diverse test oracles, we generate CAQ pairs with multiple database schemas. 
        As equivalent CAQ pairs are associated with their database schemas, we first iterate $N$ times to generate diverse schemas and their corresponding CAQs (line~\ref{line:for-loop}).
        In each iteration, we start by database seeding, which involves generating a random database schema (line~\ref{line:generate-schema}), and then producing a CAQ that conforms to the generated schema, \emph{i.e.,} the seed query $q$ (line~\ref{line:generate-caq}).
        Though existing grammar-based generators, \emph{e.g.,} SQLancer can be directly used for schema and concrete query generation, they cannot directly produce queries with placeholders. 
        The key challenge lies in producing those \modifyb{entries} in the seed CAQ while ensuring the syntactic correctness after instantiating them.
        To do so, we ask the generator to produce additional columns and tables in the schema, and then use those \emph{virtual columns} and \emph{virtual tables} to represent the two types of \modifyb{placeholders}, respectively, \emph{i.e.,} $\mathsf{Expression}$ and $\mathsf{Table}$ shown in \Cref{fig:bnf-grammar}. 
        As shown in Listing~\ref{lst:virtual-column-and-table}, we create a virtual column \texttt{placeholder1} to represent an expression, \emph{i.e.,} $\square_1 \triangleright\placeholder{\textsf{Expr}(t1:BOOLEAN)}$, and a virtual table \texttt{vtable1} for a table placeholder \emph{i.e.,} $\square_1 \triangleright\placeholder{\textsf{Table}(t0:INT, t1:INT)}$. 
        The generator, unaware of this abstraction, directly uses the concrete \texttt{placeholder1} and \texttt{vtable1}, effectively producing a valid CAQ.
        This approach makes our CAQ representation compatible with SQL provers expecting concrete syntax and easy for LLMs to understand.        

        \begin{figure}[t!]
        \centering
        \begin{lstlisting}[label={lst:virtual-column-and-table}, caption={An example of using concrete SQL query to represent CAQ when interfacing with grammar-based generators and provers.}, captionpos=b, escapeinside={(*}{*)},]
CREATE TABLE t1(c0 INT, (*\textcolor{black}{\texttt{placeholder1 \textbf{BOOLEAN}}}*));
CREATE TABLE t2(c0 INT);
CREATE TABLE (*\textcolor{black}{vtable1(c0 \textbf{INT}, c1 \textbf{INT})}*);
SELECT (*\textcolor{NavyBlue}{t1.placeholder1}*) FROM t1, (*\textcolor{NavyBlue}{vtable1}*);
\end{lstlisting}
        \vspace*{-4ex}
        \end{figure}

    \end{subsection}
    \begin{subsection}{CAQ Pairs Generation}
        \label{sec:caq-pairs-generation}
        After obtaining a seed CAQ $q$ and its schema $s$, we proceed to generate a set of equivalent CAQs $\{q_1, q_2, \ldots, q_k\}$. This is done in $\mathsf{GenerateEqual}$ (line~\ref{line:func-generate-equal} -- \ref{line:return-oracles}), which iteratively prompts an LLM $\mathcal{M}$ to produce novel variants of $q$. 
        \modifya{
            The key motivation behind using LLM here is to leverage its ability to understand the query semantics and guide the equivalent CAQ generation, while traditional grammar-based generators like SQLsmith~\cite{Seltenreich2022sqlsmith} are typically not semantics-aware and thus struggle to generate equivalent queries.
        }
        To generate variants efficiently, 
        our key strategy is to leverage in-context learning by providing the LLM with not only the seed query $q$, but also carefully selected samples from two dynamically updated sets: $\mathsf{Equal}$, \emph{i.e.,} queries previously verified as equivalent to $q$, and $\mathsf{Fail}$, \emph{i.e.,} those that failed verification. The $\mathsf{Equal}$ set promotes both correctness and diversity, while the $\mathsf{Fail}$ set provides examples of failures to avoid.

    A crucial aspect of our approach is how we guide the LLM towards generating diverse CAQs to trigger different optimization paths and detect potential bugs. 
    Instead of pursuing syntactic differences, our goal is to produce queries with varied \emph{query plans}, as this is a key principle to creating effective test oracles in previous works~\cite{ba2023testing,zhang2025constant, ba2024keep}. Qualitatively, we prompt the LLM to generate novel queries that are different from the provided examples \modify{(line \ref{line:prompt-llm})}. 
    In addition, we also quantitatively measure the similarity between query plans using the tree edit distance~\cite{zhang1989simple}.
    This metric is then used to guide the selection of examples for the next iteration's prompt.
    Specifically, we use k-means to cluster the CAQs in the $\mathsf{Equal}$ set based on the tree edit distance value and then sample queries from each cluster.
    This strategy ensures the LLM is consistently shown a diverse range of successful query structures, pushing it to explore novel execution plans.
    We defer the details of the clustering and sampling algorithms to Appendix~A.
    In summary, we can formalize the overall objective of the LLM's generation task as follows:
    \begin{prob}[Equivalent CAQ Generation]
        Given a DBMS $\mathcal{D}$, a schema $S$, a seed CAQ $q$, and a set $\mathsf{Equal}$ of previously verified CAQs equivalent to $q$, generate a new CAQ $q'$ such that (1) $q' \equiv q$ and (2) the difference between $q'$'s query plan and the plans of queries in $\mathsf{Equal}$ is maximized.
    \end{prob}

    After the LLM generates a candidate query $q'$, we employ a SQL equivalence decider $\mathcal{P}$ to formally verify its equivalence to the seed query $q$, \emph{i.e.,} $\mathcal{P} \vdash q \equiv q'$ (line~\ref{line:equivalence-check}). Additionally, we execute $q'$ on the DBMS $\mathcal{D}$ to ensure it is syntactically correct and runs without errors, \emph{i.e.,} $\mathcal{D} \models q'$. If both checks pass, $q'$ is added to the $\mathsf{Equal}$ set (line~\ref{line:add-to-equal} --- \ref{line:increase-threshold}); otherwise, it goes into the $\mathsf{Fail}$ set (line~\ref{line:add-to-fail} --- \ref{line:increase-fail-count}). This iterative process continues until a stopping criterion is met, such as reaching a maximum number of failed attempts (line~\ref{line:while-loop}) or hitting a predefined depth limit for iteration (line~\ref{line:check-max-depth}).

    \end{subsection}
    \begin{subsection}{Equivalence Checking}
        \label{sec:equivalence-checking}
        The CAQs generated by LLM are not guaranteed to be semantically equivalent. Hence, we use a SQL prover $\mathcal{P}$ to formally verify the equivalence between the seed CAQ $q$ and a generated candidate CAQ $q'$. 
        Since CAQs contain placeholders, they cannot be directly verified by SQL provers. Therefore, we substitute the placeholders with virtual columns and tables (as shown in Listing~\ref{lst:virtual-column-and-table}). We then pass the schema with these virtual entities to the prover and determine equivalence between the two concrete queries. This works because the virtual columns and tables behave like any concrete SQL entities of the same type during equivalence checking.
        If the prover can formally prove that $q \equiv q'$, we say the two CAQs and their schema $s$ form a valid test oracle, denoted as the triplet $(S, q, q')$, which can be used to test the target DBMS $\mathcal{D}$ in the following process.
        We leverage the prover as a sound verification tool and deliberately retain only those CAQ pairs that are formally verified. This design choice prioritizes reliability, ensuring that our resulting test oracles are free from false positives and can be used confidently in the subsequent testing phase. 
        
        Another limitation of these provers is that they typically support only a subset of SQL features. To mitigate this, we first prove the equivalence between two CAQs, which typically contain fewer SQL features than concrete queries. After that, we can instantiate those $\mathsf{PlaceholderEntry}$ in the CAQs with more complex and feature-rich expressions and tables to create concrete queries for testing. This highlights the effectiveness of \toolname's two-stage design, improving efficiency while also enriching the set of supported SQL features in the final test cases.
    \end{subsection}

\end{section}
\begin{section}{Test Cases Instantiation}
    \label{sec:test-case-derivation}
        We now describe our approach to instantiate test cases from the test oracles discovered in \Cref{sec:test-oracle-discovery}. As illustrated in \Cref{alg:test-cases-instantiation}, our approach consists of: (1) synthesizing a corpus of SQL snippets, (2) instantiating test queries by populating placeholders in the test oracles with those generated in the first step, and (3) creating database instances to execute the instantiated test queries to detect bugs. 
        \begin{algorithm}[t]
        \caption{\small Test Cases Instantiation Given Oracles}\label{alg:test-cases-instantiation}
        \small
        \KwIn{
            A DBMS $\mathcal{D}$, an LLM $\mathcal{M}$, a grammar-based generator $\mathcal{G}$, a set of test oracles $\mathcal{O}$ generated in \Cref{sec:test-oracle-discovery}, and the number of test cases to be instantiated per test oracle $K$. 
        }
        \KwOut{A set of bug reports $\mathcal{B}$.}
        \SetKwFunction{FMain}{$\mathsf{InstantiateTestCase}$}
        \SetKwProg{Fn}{Function}{:}{}
        \Fn{\FMain{$o = (s, q, q'), \mathcal{L}$} \label{line:instantiate-test-case}}{
            $q^* \leftarrow q$, $q'^* \leftarrow q'$; \\
            \ForEach{$\mathsf{Placeholder}$ $e$ $\mathsf{in}$ $q$ $\mathsf{or}$ $q'$}{ \label{line:for-each-placeholder}
                \eIf{$e:\mathsf{Expression}$}{
                    $t \leftarrow \mathsf{TableName}(e)$, $d \leftarrow \mathsf{Datatype}(e)$; \\
                    $e^* \leftarrow \mathsf{SampleExpression}(\mathcal{L}, \mathsf{type} = d)$; \label{line:sample-expression} \\
                    \ForEach{$\mathsf{col} \in \mathsf{ColumnNames}(e^*)$}{ \label{line:for-each-col}
                        $d^* \leftarrow \mathsf{Datatype}(\mathsf{col})$; \\ 
                        $e^* \leftarrow e^*[\mathsf{col} \mapsto \mathsf{SampleColumn}( t, \mathsf{type}=d^*)]$; \label{line:replace-col}\\

                    }
                }{
                    $s^* \leftarrow \mathsf{Schema}(e)$; \\
                    $e^* \leftarrow \mathsf{SampleTable}(\mathcal{L}, \mathsf{source} = s, \mathsf{target} = s^*)$; \label{line:sample-table}\\
                }
                $q^* \leftarrow q^*[e \mapsto e^*]$, $q'^* \leftarrow q'^*[e \mapsto e^*]$; \label{line:replace-queries}\\
            }
            \Return $(s, q^*, q'^*)$; \label{line:return-instantiated-test-case}
        }
        $\mathcal{L} \leftarrow \mathsf{GenerateCorpus}(\mathcal{D}, \mathcal{M}, \mathcal{G})$ \tcp*{\Cref{sec:corpus-generation}} \label{line:corpus-generation}
        $\mathcal{B} \leftarrow \emptyset$; \\
        \ForEach{$o = (s, q, q') \in \mathcal{O}$}{
        \For{$i \leftarrow 1$ \KwTo $K$}{
            $(s, q^*, q'^*) \leftarrow \mathsf{InstantiateTestCase}(o, \mathcal{D})$; \\
            $\mathsf{CreateDatabaseInstance}(s, \mathcal{D})$ \tcp*{\Cref{sec:database-instantiation}} \label{line:database-instantiation}
            \If{$\mathsf{Execute}(q^*, \mathcal{D}) \neq \mathsf{Execute}(q'^*, \mathcal{D})$}{ \label{line:db-execution}
                $\mathcal{B} \leftarrow \mathcal{B} \cup \{(s, \mathcal{D}\mathsf{.data},q^*, q'^*)\}$ \tcp*{\Cref{sec:database-instantiation}} 
            }
        }
        }
        \Return $\mathcal{B}$;

    \end{algorithm}

    \begin{subsection}{Corpus Synthesis}
        \label{sec:corpus-generation}

        The test case instantiation process begins with the synthesis of a SQL snippet corpus. This corpus provides a pool of expressions and tables used to populate the placeholders within our test oracles (line \ref{line:corpus-generation}). We pre-generate this corpus rather than querying the LLM for each placeholder individually to amortize the cost of LLM invocations and improve efficiency, as many snippets can be reused across multiple test oracles. 
        To generate an effective and diverse corpus, we employ a hybrid strategy that leverages both an LLM $\mathcal{M}$ 
        and a grammar-based generator $\mathcal{G}$. We use this hybrid approach since: (1) LLMs excel at generating complex query structures and diverse database features, while (2) grammar-based generators can systematically cover a wide range of corner values and edge cases. As the grammar-based method follows established techniques, we focus this section on our novel LLM-based approach.

        Our LLM-based method starts with a default schema, which contains a single table incorporating all data types supported by the target DBMS:
        \begin{lstlisting}[frame=none,numbers=none, belowskip=-1.5ex, aboveskip=1ex,xleftmargin=\parindent,escapeinside={(*}{*)}]
CREATE TABLE t(c1 INT, c2 BOOLEAN, c3 TEXT ...);
        \end{lstlisting}
        Using this schema, we prompt $\mathcal{M}$ to generate SQL expressions resulting in diverse datatypes from the default table. For example, we generate expressions such as \texttt{(c1 + c2)} : \texttt{INT}, \texttt{(c2 AND c3 IS NOT NULL)} : \texttt{BOOLEAN}, and \texttt{CONCAT(c3, "test")} : \texttt{TEXT}. The diversity comes from combining multiple columns with various SQL operators and functions, not just selecting individual columns. 
        These expression snippets can be reused across different test oracles by substituting the table and column names properly, which we will discuss in \Cref{sec:query-instantiation}.
        Furthermore, to guide $\mathcal{M}$ with the target DBMS $\mathcal{D}$'s features, we employ a simple yet effective \emph{documentation-augmented generation} 
        method:
        We first collect the official documentation of the target DBMS, \emph{e.g.}, function reference pages and operator descriptions, and partition it into smaller pieces, where each piece describes a specific feature, such as an aggregate function or a string operator. 
        Then we sample a small set of feature documentation and prompt $\mathcal{M}$ to generate SQL snippets with the sampled features.
        This approach ensures that the generated SQL snippets are diverse enough to cover a wide range of features.

    Note that, for table snippets, we can only reuse them in test oracles that share the same schema. 
        This is because table snippets must match the exact schema of the source tables they reference---including the specific table names, column names, and their datatypes. For instance, a table snippet that performs \texttt{JOIN t1}, \texttt{t2} can only be applied to test oracles whose schema contains both tables \texttt{t1} and \texttt{t2} with compatible structures.
        Nevertheless, we can still leverage $\mathcal{M}$ to generate table snippets for each schema without much overhead, as the number of unique schemas is usually small in practice.

    After the LLM-generated snippets, we need to validate whether they produce the expected datatype (\emph{i.e.,} for expressions) or schema (\emph{i.e.,} for tables) that we desire.
        To this end, we use runtime validation by executing the generated snippets on the target DBMS $\mathcal{D}$ and checking their actual output types or schemas.
    For example, as shown in Listing~\ref{lst:type-checking}, we can directly run those \texttt{SELECT} statements and then check the output column types or table schema. 
        \begin{figure}[t!]
        \begin{lstlisting}[label={lst:type-checking}, caption={An example of runtime validating LLM-generated snippets in DuckDB~\cite{raasveldt2019duckdb}.}, captionpos=b, escapeinside={(*}{*)},]
SELECT {(*\textcolor{NavyBlue}{\textsf{GeneratedExpr}}*)} FROM t; -- datatype
SELECT * FROM {(*\textcolor{NavyBlue}{\textsf{GeneratedTable}}*)}; -- schema \end{lstlisting}
        \vspace{-2ex}
        \end{figure}
        If the execution fails, we discard the snippet; otherwise, we record its actual output type or schema for later use in \Cref{sec:query-instantiation}.

    Finally, to further generate more complex expressions, we introduce a technique we term \emph{cross-combination}. The key idea is to substitute a column reference in one expression with another expression, provided that their datatypes are compatible. 
    For instance, as shown in Listing~\ref{lst:cross-combining}, consider a Boolean 
    expression \texttt{t.c1 + t.c1 > 100} and an integer expression \texttt{bit\_count(t.c2)}, where \texttt{t.c2} is a Boolean column. Because the first expression evaluates to a Boolean, it can replace \texttt{t.c2} in the second one. This substitution yields a new, valid composite expression: \texttt{bit\_count(t.c1 + t.c1 > 100)}. This recursive process allows us to build intricate expressions from simpler, generated snippets, which enriches our corpus and combines the strengths of LLM and grammar-based methods.
        
        \begin{figure}[t!]
        \centering
        \begin{lstlisting}[label={lst:cross-combining}, caption={An example of cross-combining expressions.}, captionpos=b, escapeinside={(*}{*)},]
(*$\mathsf{Expression}\ 1$*) =  (t.c1 + t.c1 > 100) : BOOLEAN
(*$\mathsf{Expression}\ 2$*) = bit_count(t.c2) : INT
(*$\Rightarrow$*) bit_count(t.c1 + t.c1 > 100) : INT
\end{lstlisting}
        \vspace*{-4ex}
        \end{figure}

    \end{subsection}

    \begin{subsection}{Query Instantiation}
        \label{sec:query-instantiation}
        After synthesizing a corpus of SQL snippets, we instantiate the test queries by populating the placeholders in the test oracles with concrete expressions and tables from the generated corpus.
        We detail this process in \Cref{alg:test-cases-instantiation}:
        given a test oracle $o = (s, q, q')$, we iterate through each placeholder $e$ in both queries $q$ and $q'$ (line~\ref{line:for-each-placeholder}). 
        Depending on whether the placeholder represents an expression or a table, we handle them differently.
    For an expression placeholder, we first determine its expected datatype $d$ and the table $t$ that the placeholder is defined over (from the CAQ's placeholder map). 
        Then, we randomly sample an expression $e^*$ from the corpus $\mathcal{L}$ that matches the expected datatype $d$ (line~\ref{line:sample-expression}).
        Since the columns in the sampled expression $e^*$ refer to the default table, we need to replace them with columns from the target table $t$ that have compatible datatypes.
        We achieve this by iterating through each column in $e^*$ (line~\ref{line:for-each-col}), determining its datatype, and substituting it with a randomly sampled column from table $t$ that matches the datatype (line~\ref{line:replace-col}).

        For a table placeholder, we first extract its expected schema $s^*$.
        Then, we randomly sample a table $e^*$ from the corpus $\mathcal{L}$ that matches the expected schema $s^*$ and is derived from the source schema $s$ (line~\ref{line:sample-table}).
    Specifically, there are three ways to replace a table placeholder: \emph{i.e.,} direct replacement with a table name, using a Common Table Expression (CTE), or \texttt{CREATE VIEW}. We randomly choose one of the three methods. 
        Finally, we replace the placeholder $e$ in both queries $q$ and $q'$ with the instantiated expression or table $e^*$ (line~\ref{line:return-instantiated-test-case}).
        By following this procedure, we ensure that the instantiated queries $q^*$ and $q'^*$ are syntactically valid and semantically equivalent, ready for execution on the target DBMS.

        Recall that, in \Cref{sec:equivalence-checking}, we prove the equivalence of the abstract query pairs $(q, q')$ 
        under the semantics of virtual columns and tables.
        However, this equivalence is only guaranteed when those virtual columns and tables exist; this may not hold when the queries are instantiated with concrete expressions and tables that introduce semantic misalignments.
        Therefore, to ensure that the instantiated queries $q^*$ and $q'^*$ remain equivalent, we enforce the following general constraints $\mathcal{C}$ (defined in Definition 4.2) on the sampled query snippets.
        \begin{enumerate}[leftmargin=*, itemsep=2pt]
            \item \textbf{Determinism} ($\mathsf{Table}$ and $\mathsf{Expression}$): The snippet should not contain any non-deterministic functions, such as \texttt{RANDOM()}.
            \item \textbf{Null-preserving} ($\mathsf{Expression}$): The expression should preserve a null return value when evaluated on rows only containing null values. For example, \texttt{c1 + c2} is null-preserving, while \texttt{IFNULL(c1, 0)} is not.
            \item \textbf{Empty Results-preserving} ($\mathsf{Expression}$): The expression should return an empty result set when evaluated on an empty table. For example, \texttt{sum(c1)} is not empty-preserving.
        \end{enumerate}
        The first constraint is also widely used in previous works~\cite{rigger2020detecting,rigger2020testing,rigger2020finding} and can be checked by applying a regular expression to identify non-deterministic features.
        \modifyc{The second and third constraints resolve the semantic misalignment that arises when virtual columns and tables are instantiated with concrete expressions and tables.
        These constraints are necessary because, in the presence of outer joins, virtual columns may no longer faithfully capture the semantics of concrete expressions after instantiation. \emph{e.g.,} outer joins can introduce \texttt{NULL} values into virtual columns, causing instantiated expressions to be evaluated differently from the version fed to the SQL equivalence prover.

        After satisfying these constraints, we can ensure that the instantiated queries $q^*$ and $q'^*$ remain equivalent, proven by contradiction:
            Intuitively, if the two queries after instantiation were not semantically equivalent, then there would exist a concrete database instance on which they produce different results.
    Using this instance, we can populate the virtual tables and columns in the prover’s schema with the corresponding values.
    After that, we obtain a database instance under the schema used in the prover where the two queries return different results, which contradicts the original semantic equivalence.
        A detailed analysis of the counterexamples, together with the full proof for the cases satisfying them, is provided in Appendix~B.
        }

        Furthermore, we validate these constraints at runtime: we execute each snippet on $\mathcal{D}$ and check the returned types or schemas and behavior on \texttt{NULL} or empty inputs. 

\end{subsection}
    
    \begin{subsection}{Database Instantiation and Bug Reporting}
        \label{sec:database-instantiation}
        After instantiating the test queries, we create the database instances to execute them (\modifyb{line~\ref{line:database-instantiation} and \ref{line:db-execution}}).
        We basically follow the widely-used random data generation method~\cite{rigger2020detecting,rigger2020testing,rigger2020finding} to populate the tables in the instantiated schema $s$ with random tuples, and then create random indices on those tables to diversify the execution plans.

        Finally, we execute the instantiated queries $q^*$ and $q'^*$ on the target DBMS $\mathcal{D}$ and compare their results (line~\ref{line:database-instantiation}).
        If the results differ, we report a logic bug along with the instantiated test case, database schema, and table data.
        In addition, if a query causes a crash of $\mathcal{D}$, we also report a crash bug.
        Note that, though our test oracles can also detect performance issues similar to~\cite{jung2019apollo, liu2022automatic}, we do not detect them in a large scale due to those issues are typically regarded as expected behaviors by developers~\cite{ba2024cert}.

    \end{subsection}
\end{section}
\section{Evaluation}

    \setstcolor{black}
    \begin{table}[t!]
        \resizebox{0.85\linewidth}{!}{%
        \begin{tabular}{l l r r r}
        \hline
        \multirow{2}{*}{DBMS} & \multirow{2}{*}{Tested by} & \multirow{2}{*}{GitHub stars} & \multirow{2}{*}{Released} & \multirow{2}{*}{LOC} \\
                              &                            &                              &                           &                       \\ \hline
        Dolt       & \cite{zhong2025scaling}    & 19.1k & 2018 & 380k   \\
        DuckDB     & \cite{zhong2025scaling, rigger2020finding, zhang2025constant, fu2024sedar, fu2023griffin} & 32.7k & 2019 & 1,496k \\
        MySQL      & \cite{zhong2025scaling, rigger2020testing, song2025detecting, jiang2024detecting, zhang2025constant, song2024detecting, liang2022detecting, zhong2020squirrel, tang2023detecting, tang2025unveiling} & 11.7k & 1995 & 5,532k \\
        PostgreSQL & \cite{zhong2025scaling, rigger2020detecting, rigger2020testing, song2025detecting, jiang2024detecting, liang2023sequence, fu2023griffin, zhong2020squirrel} & 18.5k & 1995 & 938k   \\
        TiDB       & \cite{zhong2025scaling, ba2023testing, jiang2024detecting, zhang2025constant, song2024detecting, song2025detecting, tang2025unveiling, tang2023detecting} & 39.0k & 2016 & 1,398k \\ \hline
        \end{tabular}}
        \vspace{2ex}
        \caption{DBMSs under test by \toolname.}
        \vspace{-4ex}
        \label{tab:target-dbms}
    \end{table}

    \label{sec:evaluation}
    In this section, we evaluate \toolname on five aspects:
    \begin{enumerate}[leftmargin=*, itemsep=1pt]
    \item How many new bugs can \toolname find in real-world DBMSs, and what LLM-generated test oracles discover them? (\Cref{subsec:new-bugs}) 
    \item How does \toolname compare to other baseline tools in terms of code coverage? (\Cref{subsec:coverages})
    \item How does the number of automatically discovered oracles affect the number of unique bugs found, compared to prior work where oracles are manually designed? (\Cref{subsec:effect-oracles})
    \item How effective is \toolname's SQL equivalence prover in filtering out false positives? (\Cref{subsec:effect-sql-prover})
    \item What are the time and monetary costs required for \toolname to generate a given number of test cases? (\Cref{subsec:cost-throughput})
\end{enumerate}

    \paragraph{Testbed} We conducted all experiments using a machine with 64 cores and 128 GB memory running on Ubuntu 24.04.
    We use o4-mini~\cite{openai-o4-mini} for all LLM calls with an Azure OpenAI API service in our experiments.
    We evaluate \toolname on five widely-used DBMSs: Dolt~\cite{dolt2025}, DuckDB~\cite{raasveldt2019duckdb}, MySQL~\cite{mysql2025}, PostgreSQL~\cite{momjian2001postgresql}, and TiDB~\cite{huang2020tidb}.
    Note that, we only choose DBMSs that are actively maintained and have been comprehensively tested by at least one prior work, making new bug findings more valuable.
    The details of the target DBMS are shown in \Cref{tab:target-dbms}.

    \subsection{New Bugs and Oracles}
    \label{subsec:new-bugs}
    \paragraph{Bug report policy} We continuously ran \toolname on the five DBMSs over three months.
    We used the latest development versions and reported bugs only when they could be reproduced on their latest versions.
    To avoid rediscovering known bugs, we first carefully reviewed the open issues of the target DBMS and our previous reports. Whenever a new patch was released for one of our reports, we switched to testing the latest release of the target DBMS.
    \begin{table}[t!]
        \resizebox{0.77\linewidth}{!}{%
        \begin{tabular}{lrrrrrrr}
        \hline
        \multirow{2}{*}{DBMS} & \multicolumn{1}{c}{\multirow{2}{*}{Reported}} & \multicolumn{4}{c}{Bug status}                                                   & \multicolumn{2}{c}{Bug type}                          \\
                            & \multicolumn{1}{c}{}                          & \multicolumn{1}{c}{Fixed} & \multicolumn{1}{c}{Conf.} & \multicolumn{1}{c}{Dup.} & \multicolumn{1}{l}{Pend.} & \multicolumn{1}{c}{Logic} & \multicolumn{1}{c}{Other} \\ \hline
        Dolt                  & 19                                            & 18                        & 1                         & 0 & 0                       & 18                        & 1                         \\
        DuckDB                & 8                                             & 6                         & 0                         & 1 & 1                      & 4                         & 4                         \\
        MySQL                 & 8                                             & 0                         & 5                         & 1 & 2                      & 8                         & 0                         \\
        PostgreSQL            & 1                                             & 1                         & 0                         & 0 & 0                      & 1                         & 0                         \\
        TiDB                  & 5                                             & 2                         & 3                         & 0 & 0                       & 5                         & 0                         \\ \hline
        \textbf{Total}        & \bugsnumber                                   & \fixednumber              & 9                        & 2 & 3                       & \logicnumber              & 5                         \\ \hline
        \end{tabular}}
        \vspace{2ex}
        \caption{\toolname found bugs statistics.}
        \label{tab:bug-statistics}
        \vspace{-4ex}
    \end{table}
    \paragraph{Overall result} \modifyb{\Cref{tab:bug-statistics}} summarizes the statistics of the \bugsnumber bugs that \toolname discovered during our testing campaign.
    In total, \toolname detected \bugsnumber previously unknown bugs, of which \logicnumber are logic bugs that cause incorrect query results.
    We also detected \othernumber other types of bugs, including crashes and performance issues. While finding such bugs is not the main focus of this paper, this still highlights the effectiveness of \toolname to generate complex, concrete queries for testing.
    \modifyc{Among the reports, \confirmednumber bugs have been confirmed by the developers; \fixednumber have already been fixed, while 9 are still in progress.} The remaining 2 bugs are duplicates: after developers reproduced and analyzed them, they found that these bugs shared the same root cause as some previously unfixed bugs.
    The \logicnumber logic bugs we found underscore the effectiveness of our LLM-generated test oracles. 
    Compared with recent studies that use manually designed test oracles for bug finding (\emph{e.g.,} reporting 21~\cite{ba2024keep}, 24~\cite{zhang2025constant}, and 35~\cite{jiang2024detecting} logic bugs), 
    \toolname finds more logic bugs in comparison, even though the DBMSs under test have already been extensively tested by prior work using manually crafted oracles that are designed for them.

    Next, we show several representative bugs discovered by \toolname in our new test oracles.

    \begin{figure}[t!]
        \begin{lstlisting}[label={lst:postgres-logic-bug}, caption={Incorrect \texttt{json} functions handling in PostgreSQL when executing \texttt{RIGHT} \texttt{JOIN}.}, captionpos=b, escapeinside={(*}{*)},]
CREATE TABLE t(c INT);
INSERT INTO t VALUES (1);
SELECT sub.c FROM (
    SELECT (*${\square_1}\triangleright\placeholder{\textsf{Expr}(t:INT)} \mapsto$*) 
    (*\textcolor{NavyBlue}{json\_array\_length(json\_array(3, 2, t.c))}*) 
    AS c FROM t
) AS sub
RIGHT JOIN t ON FALSE; -- {2} (*\faBug{}*)
SELECT sub.c FROM (
    SELECT NULL AS c FROM t
) AS sub
RIGHT JOIN t ON FALSE; -- {NULL} (*\faCheck{}*)
\end{lstlisting}
    \vspace{-3ex}
    \end{figure}

    \begin{figure}[t!]
    \begin{lstlisting}[label={lst:mysql-logic-bug}, caption={A logic bug in MySQL found by the similar test oracle in Listing~\ref{lst:postgres-logic-bug}.}, captionpos=b, escapeinside={(*}{*)},]
CREATE TABLE t(c0 INT);
INSERT INTO t VALUES (1);
SELECT * FROM t LEFT JOIN (
    SELECT MOD(5, 2) AS c0 FROM t
) AS t2 ON FALSE
WHERE t2.c0 IS NOT NULL; -- {1} (*\faBug{}*) {} (*\faCheck{}*)
\end{lstlisting}
    \vspace{-3ex}
    \end{figure}

    \paragraph{Example (1)} 
    Listing~\ref{lst:postgres-logic-bug} presents a logic bug in PostgreSQL that was detected by \toolname. 
    The test oracle leverages the null-handling semantics of \texttt{RIGHT} \texttt{JOIN}, which implies that when a join condition is \texttt{false}, all columns from the left table must be \texttt{NULL}. 
    Consequently, the first query should always yield a \texttt{NULL} row, which is identical to the second query's output, regardless of the value of the placeholder $\square_1$. 
    However, due to a bug in its handling of \texttt{json} functions, PostgreSQL erroneously returns \texttt{2}.
    This discovery highlights both the effectiveness of our novel, automatically generated test oracles and the LLM's ability to synthesize complex queries involving advanced SQL features, such as \texttt{json} functions. 
    The finding is particularly noteworthy given that PostgreSQL is one of the world's most robust DBMSs. 
    While several recent database testing studies have targeted PostgreSQL~\cite{zhong2025scaling,zhong2025testing,song2025detecting,ye2025sembug}, none reported finding new bugs, which attests to the effectiveness of our approach. 
    After we reported the issue, developers confirmed and fixed it within 24 hours. 

    We also find that prior test oracles, such as TLP~\cite{rigger2020finding}, struggle to detect this bug.
    For example, when we append a predicate to the first query, such as $P = \texttt{sub.c} > 2$, and then apply the three variants $\texttt{WHERE}\  P$, $\texttt{WHERE NOT}\  P$, and $\texttt{WHERE}\  P \texttt{ IS NOT NULL}$, the results of the three partitioned queries remain identical.
    As a result, TLP does not report a bug in this case.
    
    Moreover, a similar oracle uncovered a logic bug in MySQL (Listing~\ref{lst:mysql-logic-bug}), further demonstrating the versatility of the oracles generated by \toolname.
    Specifically, the query in Listing~\ref{lst:mysql-logic-bug} uses a \texttt{LEFT} \texttt{JOIN} combined with a \texttt{WHERE} clause to filter out rows where the right table's column is \texttt{NULL}, which should always yield an empty result set.
    However, due to a bug in MySQL's handling of the \texttt{MOD} function in this context, it incorrectly returns a row with value \texttt{1}.

    \paragraph{Example (2)} 
Listing~\ref{lst:dolt-logic-bug} shows a logic bug in Dolt, which was detected using an oracle generated by \toolname that leverages the semantics of the \texttt{EXISTS} predicate and primary key constraints. 
Specifically, the \texttt{EXISTS} in the first query predicate should always return \texttt{TRUE} if the \texttt{t.c0} is not nullable, which is guaranteed by the primary key constraint on \texttt{(c0, c1)}.
Thus, the first query should return all rows from table \texttt{t}, identical to the second query.
However, due to a bug in Dolt's handling of the \texttt{EXISTS} predicate, it erroneously duplicates all rows in the output.

    \begin{figure}
        \begin{lstlisting}[label={lst:dolt-logic-bug}, caption={\texttt{\textbf{EXISTS}} incorrectly duplicates rows in Dolt.}, captionpos=b, escapeinside={(*}{*)},]
CREATE TABLE t(c0 INT, c1 INT, PRIMARY KEY (c0, c1));
INSERT INTO t VALUES (1, 1);
INSERT INTO t VALUES (2, 2);
INSERT INTO t VALUES (2, 3);
SELECT * FROM (*${\square_1}\triangleright\placeholder{\textsf{Table}(c0:INT, c1:INT)} \mapsto$*) (*\textcolor{NavyBlue}{t}*)
WHERE EXISTS (
  SELECT 1 FROM t AS x WHERE x.c0 = t.c0 ); 
-- {(1,1), (2,2), (2,3), (1,1), (2,2), (2,3)} (*\faBug{}*)
SELECT * FROM (*${\square_1}\triangleright\placeholder{\textsf{Table}(c0:INT, c1:INT)} \mapsto$*) (*\textcolor{NavyBlue}{t}*); 
-- {(1,1), (2,2), (2,3)} (*\faCheck{}*)

\end{lstlisting}
    \vspace{-4.3ex}
    \end{figure}

\paragraph{Example (3)}
Listing~\ref{lst:dolt-logic-bug-2} shows another bug in Dolt with the handling of \texttt{LATERAL} joins in Dolt, which was detected using an oracle that leverages the semantics of \texttt{CROSS JOIN} and \texttt{LATERAL} joins. A \texttt{LATERAL} join allows the right-side subquery to reference columns from the left-side input; Dolt implements this feature as part of its SQL support.
Specifically, if a \texttt{CROSS JOIN LATERAL} is used to join a left table with a constant right table, \emph{i.e.,} $\texttt{SELECT} 1$.
, \emph{e.g.,} a $\mathsf{Table}$ placeholder $\square_1$, the result should be identical to a regular \texttt{CROSS JOIN} with the same right table, \emph{i.e.,} the first query is semantically equivalent to the second query.
However, due to a bug in Dolt's handling of \texttt{LATERAL} joins, when instantiating the right table to \texttt{(SELECT 1 AS c0) AS v}, the first query incorrectly returns an empty result set instead of the expected row.

\begin{figure}
        \begin{lstlisting}[label={lst:dolt-logic-bug-2}, caption={Incorrect handling of \texttt{\textbf{LATERAL}} joins in Dolt.}, captionpos=b, escapeinside={(*}{*)},]
CREATE TABLE t0(c0 BOOLEAN);
CREATE TABLE t1(c0 INT);
INSERT INTO t0 VALUES (TRUE);
INSERT INTO t1 VALUES (0);
SELECT v.c0, t1.c0 FROM t0
CROSS JOIN LATERAL (
    (*${\square_1}\triangleright\placeholder{\textsf{Table}(c0:INT)} \mapsto$*) (*\textcolor{NavyBlue}{(SELECT 1 AS c0) AS v}*)
) JOIN t1 ON v.c0 > t1.c0; -- {} (*\faBug{}*)
SELECT v.c0, t1.c0 FROM t0
CROSS JOIN (
    (*${\square_1}\triangleright\placeholder{\textsf{Table}(c0:INT)} \mapsto$*) (*\textcolor{NavyBlue}{(SELECT 1 AS c0) AS v}*)
) JOIN t1 ON v.c0 > t1.c0; -- {(TRUE, 0)} (*\faCheck{}*)
\end{lstlisting}
    \vspace{-3ex}
\end{figure}

\paragraph{Example (4)} 
The Listing~\ref{lst:duckdb-crash-bug} presents a DuckDB crash caused by a complex expression synthesized by \toolname. The expression, generated by an LLM, contains a CTE named \texttt{seq} that produces a sequence of integers. 
The crash occurs when this CTE is used with an aggregation function, \emph{i.e.,} \texttt{SUM}, inside an \texttt{EXISTS} predicate. Notably, synthesizing such a non-trivial expression is challenging for traditional grammar-based generators, which require both complex query structures and various SQL features.
This case demonstrates that by leveraging the generative capabilities of LLMs for SQL snippets, \toolname can effectively uncover not only logic bugs but also critical crashes.

\begin{figure}
        \begin{lstlisting}[label={lst:duckdb-crash-bug}, caption={A crash bug in DuckDB triggered by a LLM-synthesized expression with \texttt{\textbf{seq}}.}, captionpos=b, escapeinside={(*}{*)},]
CREATE TABLE t0(c0 INT);
CREATE TABLE t1(c0 BOOLEAN);
SELECT * FROM t0
WHERE EXISTS (
  SELECT 1 FROM t1
  WHERE (*${\square_1}\triangleright\placeholder{\textsf{Expr}(t0:INT)} \mapsto$*)
  (*\textcolor{NavyBlue}{(WITH seq(i) AS (VALUES (1)) SELECT sum(i) *\ t0.c0 FROM seq)}*) IS NOT NULL
); -- (crash) (*\faBug{}*)
\end{lstlisting}

\vspace{-2ex}
\end{figure}

\definecolor{colorTLO}{HTML}{cc6666}
\definecolor{colorSQL}{HTML}{6699cc}
\definecolor{colorSQLpp}{HTML}{66cc99}
\definecolor{colorEET}{HTML}{ffcc66}
\definecolor{colorSQLR}{HTML}{ff6699}
\definecolor{colorArgus}{HTML}{8A2BE2}

\begin{figure}[tp!]
    \centering
    \begin{minipage}{0.95\linewidth}
        \centering
        \begin{tikzpicture}[baseline=-0.5ex]
            \draw[colorTLO, solid, line width=1.6pt] (0,0) -- (2.8em,0) node[pos=0.5, star, star points=5, fill=colorTLO, inner sep=1.4pt] {};
        \end{tikzpicture}
        \toolname
        \quad\quad\quad
        \begin{tikzpicture}[baseline=-0.5ex]
            \draw[colorSQL, solid, line width=1.6pt] (0,0) -- (2.8em,0) node[pos=0.5, circle, fill=colorSQL, inner sep=1.5pt] {};
        \end{tikzpicture}
        SQLancer
        \quad\quad\quad
        \begin{tikzpicture}[baseline=-0.5ex]
            \draw[colorSQLpp, solid, line width=1.6pt] (0,0) -- (2.8em,0) node[pos=0.5, rectangle, fill=colorSQLpp, inner sep=1.8pt] {};
        \end{tikzpicture}
        SQLancer++
        \quad\quad\quad
        \begin{tikzpicture}[baseline=-0.5ex]
            \draw[colorEET, solid, line width=1.6pt] (0,0) -- (2.8em,0) node[pos=0.5, diamond, fill=colorEET, inner sep=1.8pt] {};
        \end{tikzpicture}
        EET
    \end{minipage}\\[1.5ex] 
    \centering
    \begin{minipage}{1\linewidth}
        \centering
        \begin{minipage}{0.48\linewidth}
            \centering
            \includegraphics[width=\linewidth]{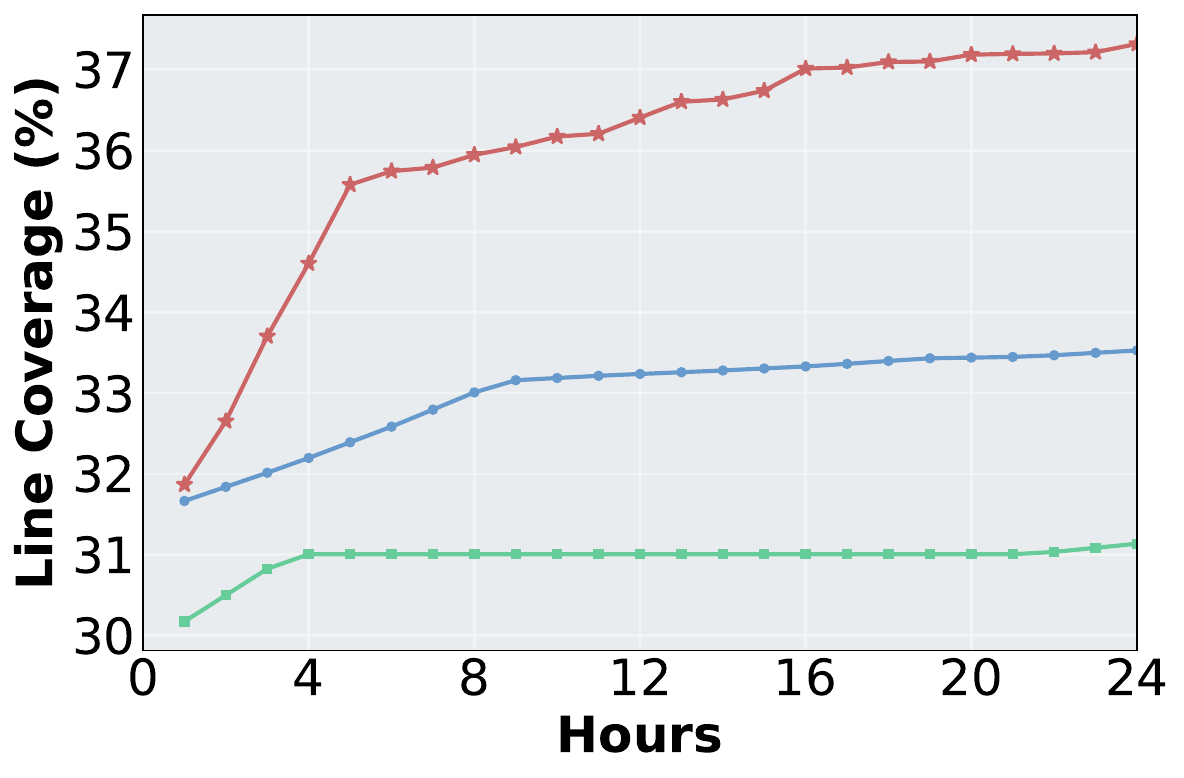}\\[-0.4ex]
            \scriptsize (a) DuckDB line
        \end{minipage}\hfill
        \begin{minipage}{0.48\linewidth}
            \centering
            \includegraphics[width=\linewidth]{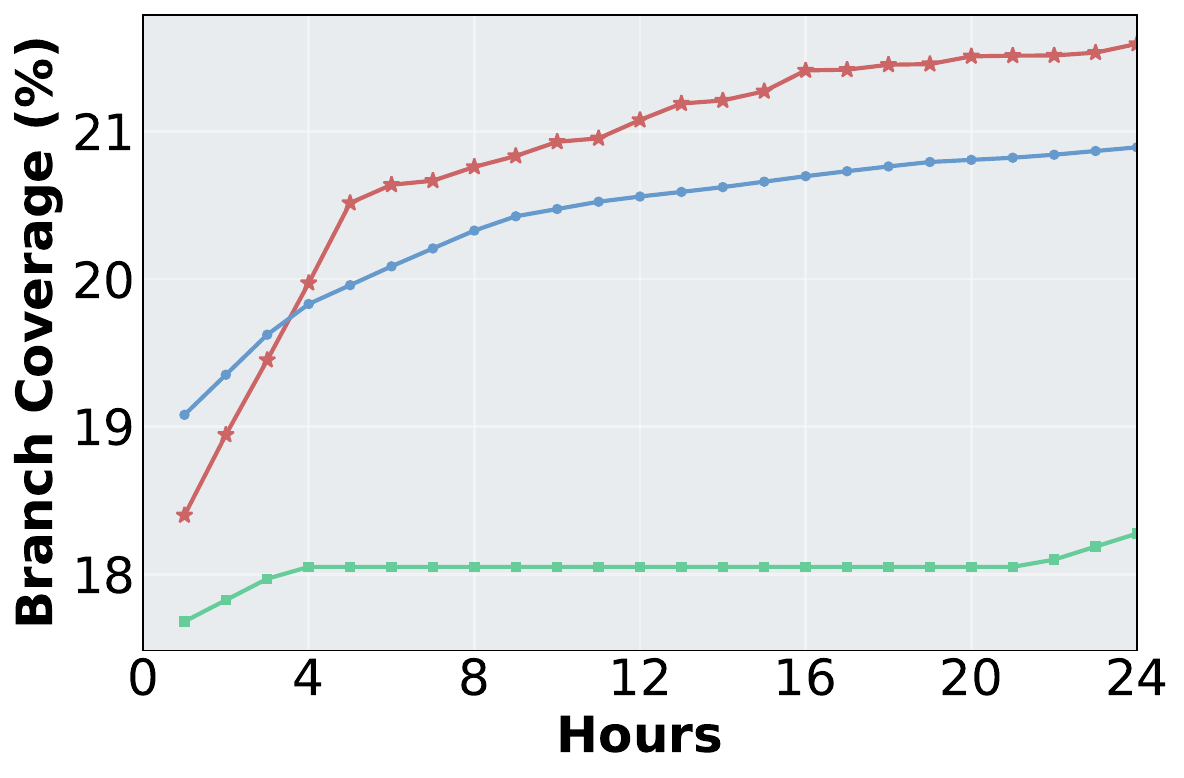}\\[-0.4ex]
            \scriptsize (b) DuckDB branch
        \end{minipage}\\[1ex]
        \begin{minipage}{0.48\linewidth}
            \centering
            \includegraphics[width=\linewidth]{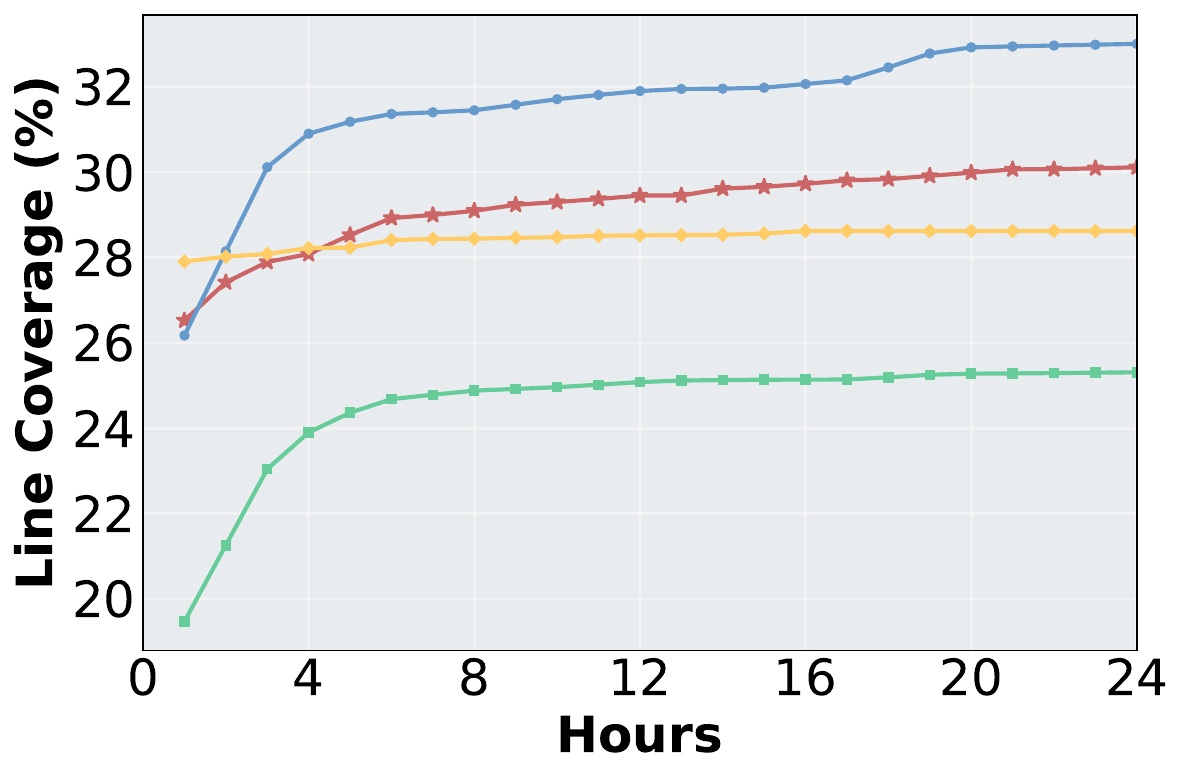}\\[-0.4ex]
            \scriptsize (c) PostgreSQL line
        \end{minipage}\hfill
        \begin{minipage}{0.48\linewidth}
            \centering
            \includegraphics[width=\linewidth]{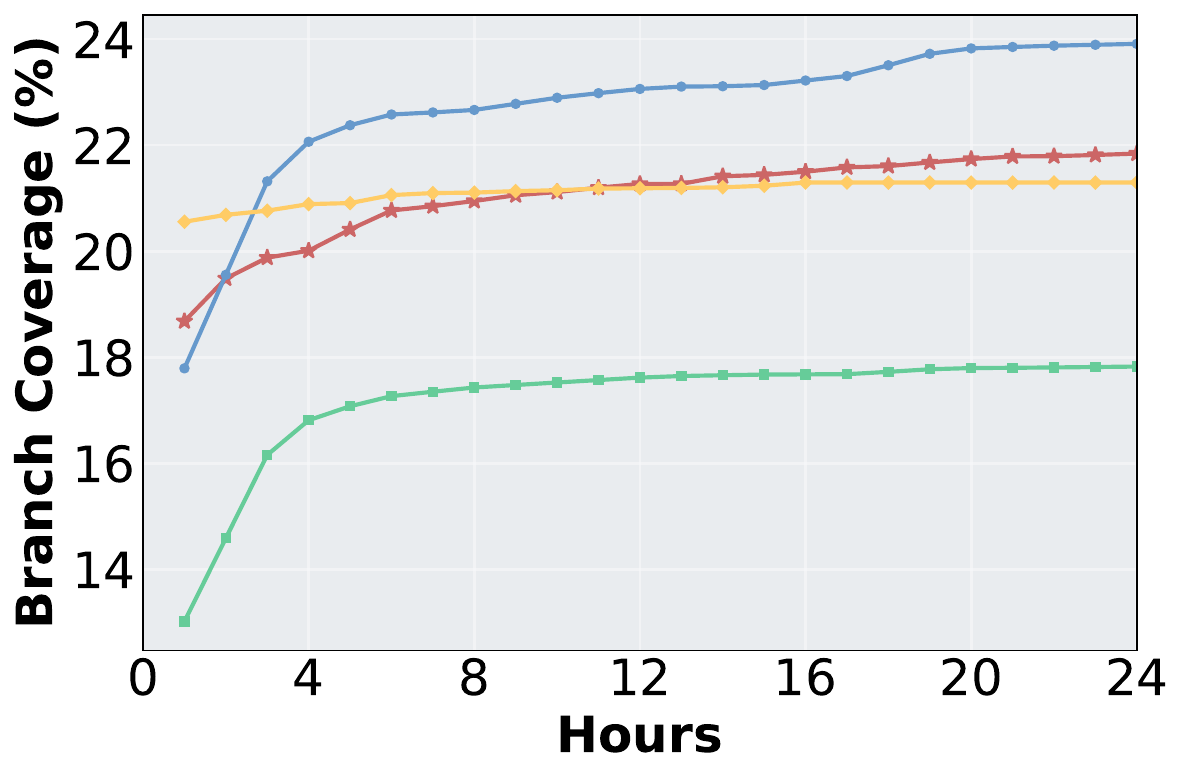}\\[-0.4ex]
            \scriptsize (d) PostgreSQL branch
        \end{minipage}
    \end{minipage}
    \vspace{1ex}
    \caption{\modifyb{Code coverage achieved by \toolname, SQLancer, and SQLancer++ on DuckDB and PostgreSQL over 24-hour runs.}}
    \label{fig:code-coverage}
\end{figure}

\subsection{Code Coverage}
\label{subsec:coverages}
We compared \toolname with three DBMS logic bugs finding tools, SQLancer~\cite{ba2023testing,rigger2020finding,rigger2020testing,rigger2020detecting,zhang2025constant}, SQLancer++~\cite{zhong2025scaling} and EET~\cite{jiang2024detecting} in multiple coverage metrics.
We compared on DuckDB~\cite{raasveldt2019duckdb} and PostgreSQL~\cite{momjian2001postgresql}.
As a state-of-the-art, open-source DBMS testing framework, SQLancer supports most of the latest test oracles~\cite{ba2024cert,zhang2025constant,song2024detecting}.
SQLancer++ is a scalable variant that can be easily extended to multiple DBMSs with only lightweight modifications.
EET~\cite{jiang2024detecting} is a recent tool based on SQLsmith~\cite{Seltenreich2022sqlsmith}'s generator and their carefully designed test oracles to produce complex queries.

In terms of test oracles, \emph{code coverage} is an indicator for the number of features and execution plans that are tested by them. 
Although code coverage does not strongly correlate with the ability to find logic bugs~\cite{zhong2025scaling}, it is still a fair metric to evaluate the diversity of test cases generated by \toolname and the effectiveness of its corpus synthesis. 
Meanwhile, \emph{metamorphic coverage}~\cite{ba2025metamorphic} is a more relevant metric for evaluating the effectiveness of test oracles and is highly related to logic bug finding abilities according to the historical bug study in SQLite and DuckDB~\cite{ba2025metamorphic}.
\modifyb{
Unlike traditional code coverage, which simply counts all executed code, metamorphic coverage measures the portions of the program exercised differently by the two executions in a pair of equivalent queries.
Intuitively, because only these differential execution paths can lead to inconsistent results between semantically equivalent queries, metamorphic coverage provides a more precise measure of the code actually validated by the test oracle.
Empirical studies on DBMSs like SQLite and DuckDB show that metamorphic coverage correlates strongly with real logic bugs and overlaps substantially with historical bug-fix locations.
}

\vspace{-2ex}
\paragraph{Code coverage}
\Cref{fig:code-coverage} shows the line and branch code coverage for DuckDB and PostgreSQL over a 24-hour testing period starting from a clean build and empty database. 
In DuckDB, \toolname achieves 19.9\% and 18.1\% higher line and branch coverage than SQLancer++, respectively, and 11.3\% and 3.4\% higher than SQLancer.
Note that, we did not compare with EET on DuckDB, as it does not support this DBMS.
In PostgreSQL, \toolname slightly underperforms SQLancer overall but outperforms SQLancer++ (by 19.0\% in line coverage and 22.5\% in branch coverage) and EET (by 5.2\% and 2.6\%, respectively).
This suboptimal performance is expected, as SQLancer has been extensively optimized specifically for PostgreSQL by the open-source community over many years.
Specifically, SQLancer supports 22 types of DDL statements (e.g., \texttt{CREATE SEQUENCE}) and various DML statements beyond \texttt{SELECT} queries that PostgreSQL supports, whereas \toolname focuses on detecting logic bugs across different DBMS implementations.

For a finer-grained comparison, we additionally evaluated \toolname and SQLancer on PostgreSQL with respect to optimizer code coverage, since the optimizer is a core component that is closely tied to \texttt{SELECT} queries and widely studied in prior work~\cite{tang2023detecting, rigger2020detecting}, along with the diversity of features exercised by \texttt{SELECT} queries that are covered by the test queries.
The results show that \toolname achieved 66.46\% line coverage and 55.85\% branch coverage of PostgreSQL’s optimizer, compared to 62.33\% and 52.17\% by SQLancer.
To evaluate feature diversity, we randomly sampled ten \texttt{SELECT} queries from each tool and counted their unique features using two third-party parsers: \emph{pglast}~\cite{pglast} (counting number of unique PostgreSQL types) and \emph{sqlparse}~\cite{sqlparse} (counting number of unique AST nodes).
We found that \toolname covered 23 features in pglast and 151 features in sqlparse, while SQLancer only covered 15 and 76, respectively.
This demonstrates \toolname's stronger ability to generate feature-rich queries.
The sampled queries are provided in~Appendix~D.


\begin{table}[t!]
    \centering
    \begin{tabular}{lrrr}
    \hline
    \multirow{2}{*}{Approach} & \multirow{2}{*}{Lines} & \multirow{2}{*}{Functions} & \multirow{2}{*}{Branches} \\
                            &                        &                            &                           \\ \hline
    SQLancer                  & 3.256\%                & 1.230\%                    & 1.313\%                   \\
    \toolname                 & \textbf{17.820\%}      & \textbf{7.910\%}           & \textbf{7.315\%}          \\
                              & \textbf{\textcolor{OliveGreen}{5.473$\times$}} & \textbf{\textcolor{OliveGreen}{6.431$\times$}} & \textbf{\textcolor{OliveGreen}{5.571$\times$}} \\ \hline
    \end{tabular}
    \vspace{2ex}
    \caption{Average metamorphic coverage on DuckDB of 10 test suites for \toolname and SQLancer.}
    \vspace{-4ex}
    \label{tab:metamorphic-coverage}
\end{table}

Note that we did not compare with mutation-based testing tools, such as SQLRight~\cite{liang2022detecting} and SQuirrel~\cite{zhong2020squirrel}, as they highly rely on the quality of seed queries and use the official unit tests provided by the DBMS as the initial input. 
This introduces bias in the comparison, as the official unit tests are usually well-designed and purposefully cover many code paths for the DBMSs that they are designed for. In fact, SQLRight achieved $51.3\%$ line coverage on PostgreSQL in the first few minutes after feeding PostgreSQL's official unit test cases as seed queries, but quickly saturated after that without further improvement during the remainder of the testing period.

Overall, our results demonstrate that \toolname can generate feature-rich, diverse, and complex test cases at scale, covering more optimizer paths than state-of-the-art techniques.

\paragraph{Metamorphic coverage}
For metamorphic coverage, we compared \toolname with SQLancer on DuckDB across a fixed number of test cases, following the same setting in~\cite{ba2025metamorphic}.
We did not evaluate metamorphic coverage on PostgreSQL, which has not been supported by ~\cite{ba2025metamorphic}'s implementation.
Following the experimental setup in~\cite{ba2025metamorphic}, we generated 10 test suites for SQLancer, each with 100 test cases. 
Specifically, SQLancer produce 50\% of the test cases from TLP~\cite{rigger2020finding} and 50\% from NoREC~\cite{rigger2020detecting}.
For \toolname, we generated 10 CAQ pairs as test suites, and instantiated each pair with 100 different snippets, respectively.

We reported the average metamorphic coverage of the 10 test suites in \Cref{tab:metamorphic-coverage}.
As shown, \toolname achieves 5.473$\times$, 6.431$\times$, and 5.571$\times$ more line, function, and branch coverage than SQLancer, respectively.
We believe this is because \toolname's new test oracles can significantly alter the query structures between equivalent queries, thanks to LLM's creativity, which covers a broader range of different code paths than prior works.
This significant improvement demonstrates the effectiveness of \toolname's new test oracles in comparing different execution paths, and the potential to detect more logic bugs that state-of-the-art testers missed.


\subsection{Effect of Test Oracles}
\label{subsec:effect-oracles}

\begin{figure}[t!]
    \centering
    \includegraphics[width=0.5\linewidth]{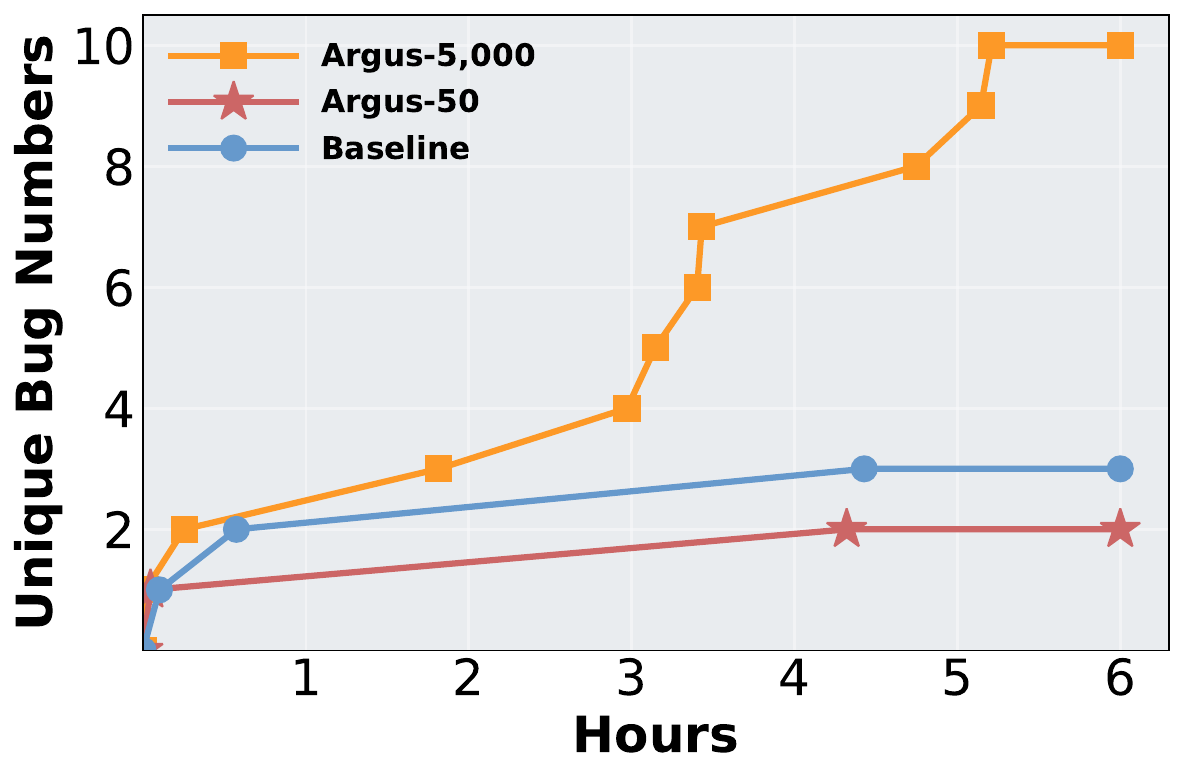}
    \vspace{-0.2in}
    \caption{Number of unique logic bugs found by different sets of oracles within a 6-hour testing on Dolt.}
    \label{fig:unique-bugs-vs-oracles}
    \vspace{-2ex}
\end{figure}

We now evaluate \toolname's new oracles by comparing the number of unique logic bugs they detected in Dolt v1.0.0. We selected this historical version to ensure a fair comparison, as \texttt{git-bisect} allows us to precisely link each bug to its corresponding fix commit, providing an accurate count of unique bugs per method.

We compared a composite baseline of test oracles against two sets of LLM-generated test oracles by \toolname, with sizes of 50 and 5,000, \emph{i.e.,} \toolname-50 and \toolname-5,000, respectively. 
The baseline combines 11 test oracles from four previous works: TLP~\cite{rigger2020finding}, NoREC~\cite{rigger2020detecting}, EET~\cite{jiang2024detecting}, and DQP~\cite{ba2024keep}.
To minimize confounding factors, we standardized the experimental conditions. 
Specifically, the baseline oracles were represented in the same CAQ pair format and instantiated using the same snippet corpus as \toolname. Furthermore, since the CAQ format requires an associated schema, we also randomly generated a new schema for the baseline oracles before each instantiation.
The details about how to convert the baseline oracles to CAQ pairs are provided in~Appendix~C.

\Cref{fig:unique-bugs-vs-oracles} shows the number of unique logic bugs found by different sets of oracles within a 6-hour testing. 
We select a 6-hour time window for a fair comparison, following common practice in fuzz testing~\cite{klees2018evaluating} and prior studies on DBMS logic bug detection~\cite{zhong2025testing}, which typically fix the testing duration within a range of 1 -- 24 hours.
As shown, \toolname-5,000 found 10 unique bugs, significantly outperforming the baseline, which found only 3. 
In contrast, \toolname-50 found just 2 bugs. 
This underperformance was expected, as the baseline oracles were carefully designed by human experts. 
These results demonstrate the importance of the number of oracles in detecting unique logic bugs and underscore the limitations of manual oracle design, thus highlighting the necessity of \toolname's automated test oracle discovery. 

\begin{figure}[t!]
    \begin{lstlisting}[language=SQL, label={lst:false-positive-example}, caption={A false positive example from LLM-as-a-judge.}, captionpos=b]
SELECT * FROM v WHERE TRUE; 
SELECT * FROM v WHERE v.c >= v.c OR v.c < v.c
\end{lstlisting}
\vspace{-4ex}
\end{figure}

\subsection{Effect of SQL Equivalence Prover}
    \label{subsec:effect-sql-prover}
We evaluate our SQL equivalence prover's false positive and false-negative rates against an LLM-as-a-judge baseline. For this baseline, we use GPT-5 to determine whether a candidate CAQ pair is equivalent directly.

\paragraph{False positive bug reports}
To quantify the impact of the prover on false positives, we performed an ablation that disables the prover while using an LLM-as-a-judge approach~\cite{zhao2023llm}, and ran \toolname on DuckDB until we collected 20 bug reports.
We then manually adjudicated each report as a true positive or a false positive.
Surprisingly, we found that \textbf{all} of them were false positives.
For instance, LLM-as-a-judge incorrectly decides the two queries shown in Listing~\ref{lst:false-positive-example} as equivalent, which is actually inequivalent when $v.c$ is \texttt{NULL}.

To further investigate the issue of extremely high false positive rates, we generated 20 candidate CAQ pairs that the LLM judged as equivalent, and manually verified them.
We found that 1 of these 20 pairs was an incorrect test oracle, \emph{i.e.,} it consists of inequivalent CAQ pairs.
While the accuracy of LLM-as-a-judge is high, bugs in mature DBMSs are exceedingly rare, as it often takes thousands of queries to discover a single bug.
Consequently, even a low false positive rate can generate a volume of false reports that overwhelmingly drown out true positives.
In addition, such false positives are particularly unacceptable when \toolname is integrated into CI/CD pipelines or production environments, as developers have limited time to investigate each report.
Therefore, the SQL equivalence prover is necessary to ensure the soundness of \toolname's test oracles and make it practical for real-world use.

\modifyc{
    In terms of the performance with the SQL equivalence prover, we sample 20 CAQ pairs that are accepted by the prover as equivalent, as well as 20 pairs that are instantiated from them.
    After our manual validation, we find that all 20 CAQ pairs and concrete instances derived from them are indeed equivalent.
    This demonstrates the prover's effectiveness in identifying and validating equivalent queries and reducing false positives.
}

\paragraph{False negatives in test-oracle generation}
To estimate the prover's false-negative rate, \emph{i.e.,} its inability to prove correct test oracles, we ran two configurations of \toolname on DuckDB: one using the prover to check CAQ pairs, and one using LLM-as-a-judge.
In each configuration, we continued testing until the tool had rejected 10 candidate CAQ pairs as ``inequivalent,'' and then we manually adjudicated each pair. 
We found that 8 of the 10 pairs were false negatives for the prover, while 5 of the 10 pairs were false negatives for the LLM-as-a-judge.
The results also show that the prover's false negative rate is comparable to that of state-of-the-art LLMs.
Note that, false negatives from the prover are expected in general, as it is intentionally conservative to ensure soundness.

\subsection{Cost and Efficiency Analysis}
\label{subsec:cost-throughput}

We next evaluate CAQ's impact on cost and efficiency by comparing \toolname against the naive baseline introduced in \Cref{sec:introduction}. This baseline uses the same SQL equivalence prover but directly prompts an LLM to generate entire pairs of equivalent queries from the seed queries rather than CAQ pairs.
After testing on Dolt 1.0.0 for one hour, \toolname detected a logic bug while the baseline found none, and we measured the corresponding test cases generated and LLM costs.

\begin{figure}[t!]
    \centering
    \includegraphics[width=0.6\linewidth]{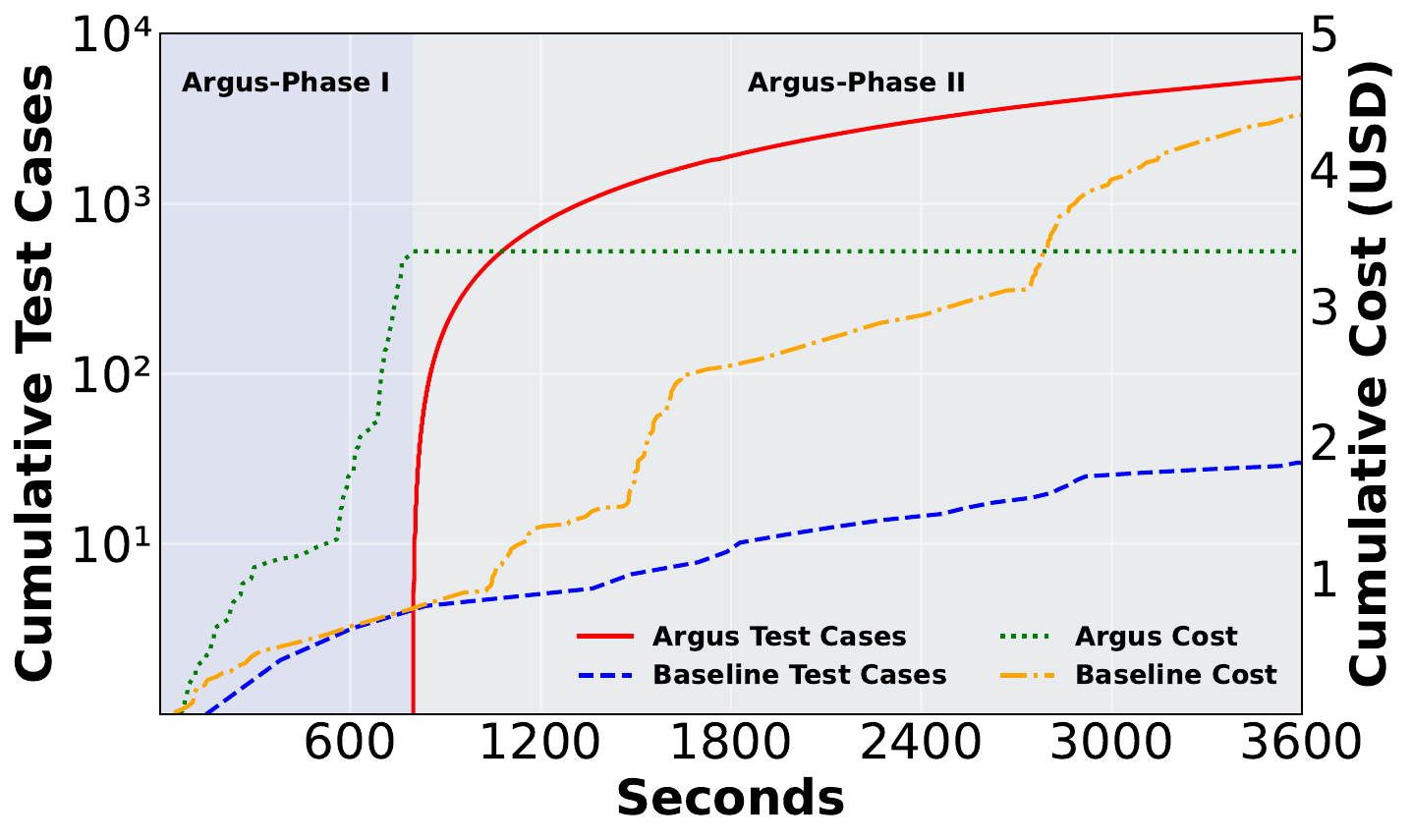}
    \vspace{-.2in}
    \caption{\modifyb{Cost and efficiency comparison between \toolname and the naive baseline on Dolt over a 1-hour run.}}
    \label{fig:cost-throughput}
    \vspace{-2ex}
\end{figure}
As shown in \Cref{fig:cost-throughput}, \toolname achieves a significantly higher test case throughput than the naive baseline. Unlike the baseline, which slowly generates complete equivalent pairs, \toolname spends only a few minutes creating abstract CAQ pairs (Phase I) and then rapidly instantiates them into numerous concrete test queries (Phase II). 
This process is also highly economical: generating the CAQ pairs costs about \$3 in LLM calls, and instantiating test cases costs roughly \$1 per 1,000 instantiated test cases, and leveraging a reusable corpus of 100,000 snippets built for only \$12. This shows that \toolname is a cost-effective solution for detecting logic bugs with a high throughput. 

\subsection{Component-wise Analysis}
\modifyb{
We analyze the contributions of two \toolname's key components: the diversity-oriented equivalent query generation (\Cref{sec:caq-pairs-generation}) and the LLM-powered SQL snippets generation (\Cref{sec:corpus-generation}).

\paragraph{Diversity-oriented equivalent CAQ generation}
To evaluate the impact of diversity-oriented equivalent query generation, we compare \toolname’s CAQ generation method against two baselines: (i) a naive approach that independently prompts an LLM to generate equivalent queries, and (ii) a beam-search–based method that produces multiple candidates at each step and selects a diverse subset, according to a diversity metric, to condition the next iteration.
}
\modifyb{
For all methods, we give the same seed CAQs, fix the LLM budget to 1,000 generated CAQ candidates and evaluate: (1) the number of valid queries (\emph{i.e.,} those that pass the SQL equivalence checker); and (2) the similarity of the generated queries, measured by the average tree-edit distance between DuckDB query plans across all pairs of equivalent queries. Specifically, given two query plan tree $T_1$ and $T_2$, let $\text{dist}(T_1, T_2)$ be the tree-edit distance between them.
The similarity score is computed as: $\frac{2}{N(N-1)} \sum_{i \neq j} \big(1 - \frac{\text{dist}(T_i, T_j)}{|T_i| + |T_j|} \big)$, where $N$ is the number of valid query pairs.
As shown in \Cref{tab:diversity-comparison}, \toolname’s CAQ generation method outperforms both baselines by producing more valid CAQs while maintaining the lowest similarity.
}
\begin{table}[t!]
    \centering
    \begin{tabular}{lrr}
    \hline
    Method & Avg. Similarity & \# of Valid Queries \\
    \hline
    Clustering  & \textbf{0.617} & \textbf{86} \\
    Beam-search & 0.766          & 48 \\
    Naive       & 0.688          & 81 \\
    \hline
    \end{tabular}
    \vspace{2ex}
    \caption{\modifyb{Comparison of \toolname's equivalent CAQ generation methods with baselines. Lower similarity indicates higher diversity.}}
    \vspace{-4ex}
    \label{tab:diversity-comparison}
\end{table}

\paragraph{LLM-powered SQL snippet generation}
\modifyb{
    We evaluate the effectiveness of LLM-powered SQL snippet generation by comparing \toolname's method against a variant that only uses SQLancer's grammar-based generator to build the snippet corpus during the CAQ instantiation phase. We use the code coverage metric in \Cref{subsec:coverages} to evaluate both methods after a 24-hour testing on DuckDB.
    As shown in \Cref{fig:code-coverage-ablation}, \toolname consistently outperforms the grammar-based variant in both line and branch coverage.
    Without LLM-generated snippets, the grammar-based variant can still outperform baselines in terms of line coverage but struggles with branch coverage.
}

\begin{figure}[t!]
    \centering
    \begin{minipage}{0.95\linewidth}
        \centering
        \begin{tikzpicture}[baseline=-0.5ex]
            \draw[colorTLO, solid, line width=1.6pt] (0,0) -- (2.8em,0) node[pos=0.5, star, star points=5, fill=colorTLO, inner sep=1.4pt] {};
        \end{tikzpicture}
        \toolname
        \quad\quad\quad
        \begin{tikzpicture}[baseline=-0.5ex]
            \draw[colorSQL, solid, line width=1.6pt] (0,0) -- (2.8em,0) node[pos=0.5, circle, fill=colorSQL, inner sep=1.5pt] {};
        \end{tikzpicture}
        SQLancer
        \\
        \begin{tikzpicture}[baseline=-0.5ex]
            \draw[colorSQLpp, solid, line width=1.6pt] (0,0) -- (2.8em,0) node[pos=0.5, rectangle, fill=colorSQLpp, inner sep=1.8pt] {};
        \end{tikzpicture}
        SQLancer++
        \quad\quad\quad
        \begin{tikzpicture}[baseline=-0.5ex]
            \draw[colorArgus, solid, line width=1.6pt] (0,0) -- (2.8em,0) node[pos=0.5, regular polygon, regular polygon sides=3, fill=colorArgus, inner sep=1pt] {};
        \end{tikzpicture}
        Argus-grammar-based
    \end{minipage}\\[1.5ex] 
    \centering
    \begin{minipage}{0.95\linewidth}
        \centering
        \begin{minipage}{0.48\linewidth}
            \centering
            \includegraphics[width=\linewidth]{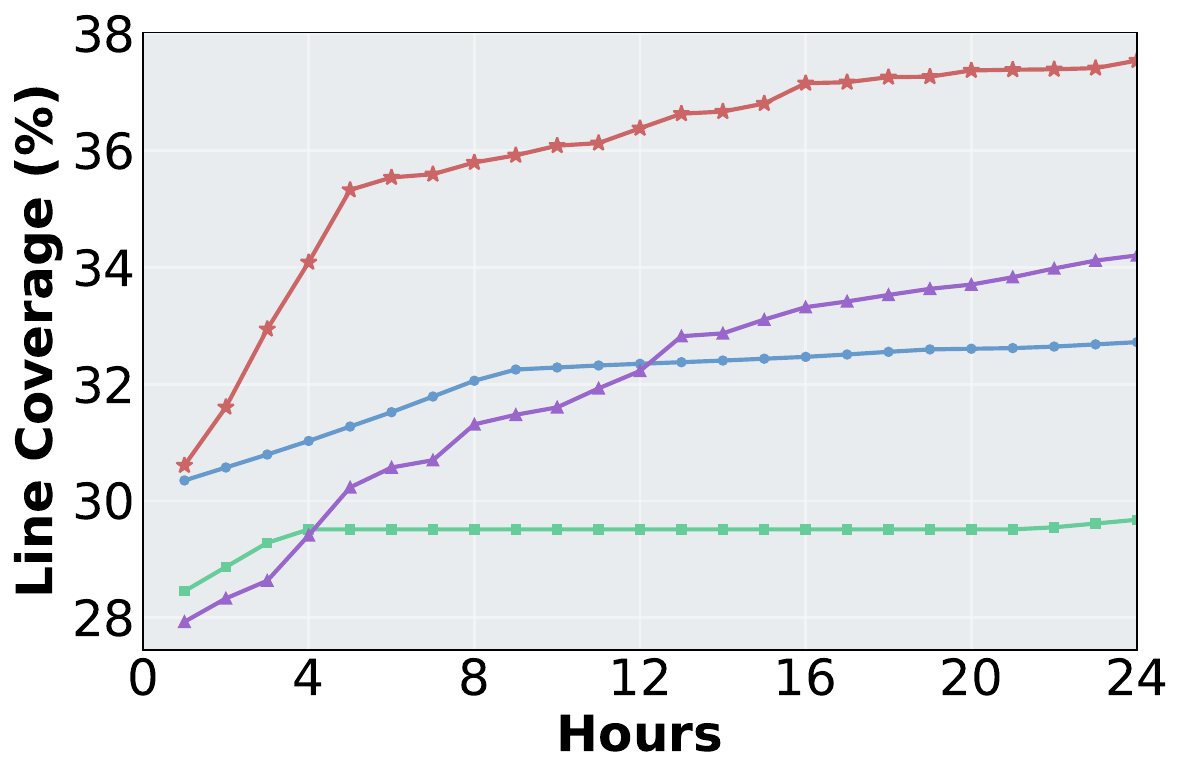}\\[-0.4ex]
            \scriptsize (a) DuckDB line
        \end{minipage}\hfill
        \begin{minipage}{0.48\linewidth}
            \centering
            \includegraphics[width=\linewidth]{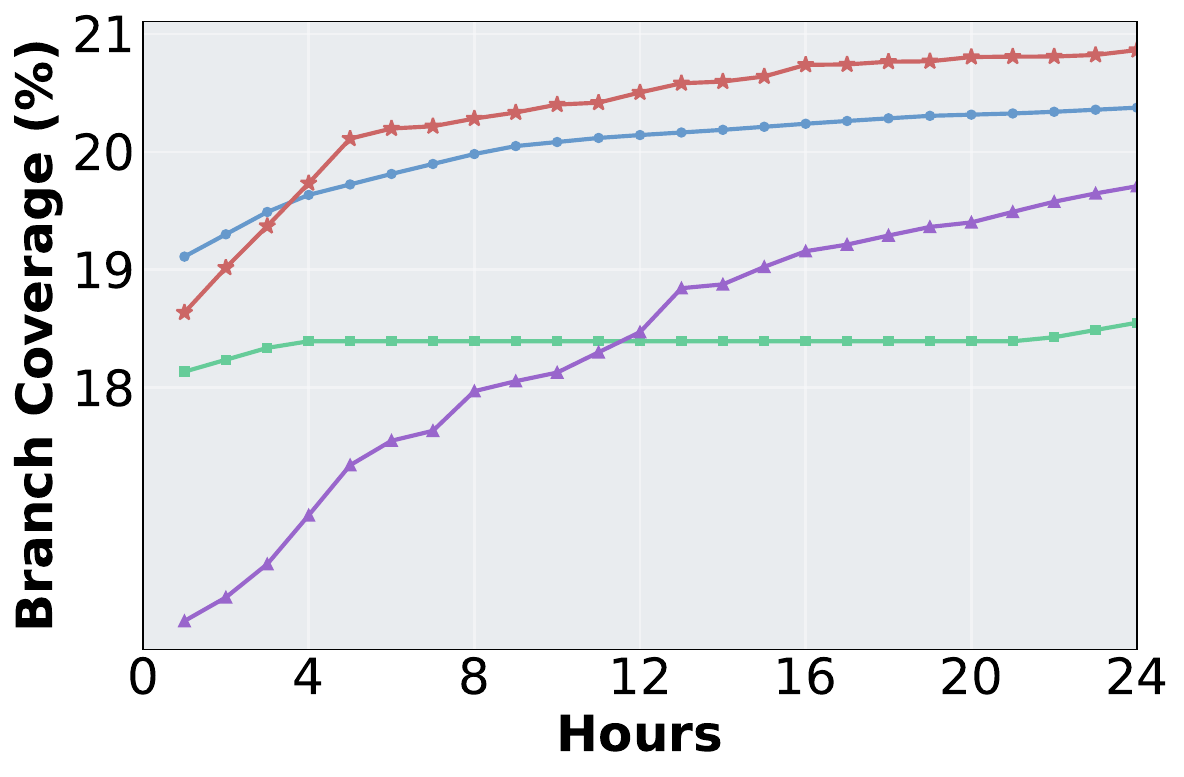}\\[-0.4ex]
            \scriptsize (b) DuckDB branch
        \end{minipage}
    \end{minipage}
    \vspace{1ex}
    \caption{\modifyb{Ablation study of code coverage achieved by \toolname and a variant using only SQLancer’s grammar-based generator on DuckDB over 24-hour runs.} }
    \label{fig:code-coverage-ablation}
\end{figure}

\section{Discussion}

\paragraph{Additional finding}
\begin{figure}[t!]
\centering
\begin{lstlisting}[label={lst:sqlsolver-bug}, caption={Incorrect equivalence proof in SQLSolver~\cite{ding2023proving}.}, captionpos=b, escapeinside={(*}{*)},]
CREATE TABLE t1(c0 BOOLEAN, c1 BOOLEAN);
SELECT false FROM t1;
SELECT (t1.c1 < t1.c0) < (t1.c1 = t1.c1) FROM t1;
\end{lstlisting}
\vspace{-3ex}
\end{figure}
In addition to improving DBMS robustness, our testing also revealed several bugs in existing SQL equivalence provers. These bugs could lead to inherent false positives that, if unaddressed, would hinder large-scale testing.
For example, an early version of SQLSolver~\cite{ding2023proving} incorrectly proved the two queries in Listing~\ref{lst:sqlsolver-bug} to be equivalent: the first query always returns \texttt{false} when \texttt{t1} is non-empty, whereas the second query can return \texttt{true} or even \texttt{NULL} depending on the values of \texttt{c0} and \texttt{c1}.
Such issues are common among research prototypes that have been validated only on small benchmarks --- for instance, the widely used Calcite test cases~\cite{begoli2018apache} include only about 200 equivalent query pairs.
We reported 10 bugs in SQLSolver and QED~\cite{wang2024qed} during early development of \toolname, all of which were quickly fixed by the developers. After these fixes, our testing proceeded without being affected by those prover bugs.
These findings demonstrate that our approach not only helps uncover DBMS bugs but also contributes to improving the reliability of SQL equivalence provers. Importantly, the soundness of our approach is independent of the correctness of specific prover implementations—the latter only affects practical false positives before fixing.

\paragraph{Customized testing}
\modifya{
    \toolname is highly configurable and supports a broad range of testing scenarios. By adjusting the LLM generation prompt, users can focus on specific SQL features without modifying the underlying code, enabling flexible and comprehensive testing.
    For example, to generate table snippets involving \texttt{OUTER}~\texttt{JOIN}s, the prompt can explicitly request: 
    ``Ensure that the generated snippet includes at least one \texttt{OUTER}~\texttt{JOIN} (\emph{e.g.,} \texttt{LEFT}~\texttt{JOIN}, \texttt{RIGHT}~\texttt{JOIN}, \texttt{FULL}~\texttt{OUTER}~\texttt{JOIN}, $\ldots$)''.
}

\paragraph{Query classes supported by SQL equivalence provers}
\modifyc{
    SQL equivalence provers is a key component of \toolname, but typically only support a limited set of SQL features due to the inherent complexity of SQL semantics.
    To the best of our knowledge, existing SQL equivalence provers support the core syntax implemented in Calcite~\cite{begoli2018apache}, including outer joins, nested queries, and simple aggregations like \texttt{SUM} and \texttt{COUNT}.
    However, they generally do not understand the semantics of most DBMS-specific or user-defined functions, treating them as uninterpreted functions, which may result in omitting some potential test oracles.
    For example, if a DBMS defines a function \texttt{add}$(x, y) := x + y$, the prover cannot determine whether \texttt{SELECT} \texttt{add}$(x, y)$ \texttt{FROM} $t$ is equivalent to \texttt{SELECT} $x + y$ \texttt{FROM} $t$.
    Another limitation stems from the fact that SQL equivalence provers heavily rely on SMT solvers, whose complexities can grow rapidly with the size of the CAQ pairs analyzed. 
    This may limit the query sizes of the test oracles that \toolname can generate.
    Nevertheless, these restrictions from SQL equivalence provers can be mitigated in \toolname's test cases instantiation phase, which can generate complex SQL queries by filling in CAQ placeholders with diverse SQL snippets.
    Those snippets can include SQL features that currently lack support in existing provers, as well as increase the overall complexity of the generated test cases.
}

\paragraph{Limitations} 
Our approach has two limitations that present opportunities for future work.
\begin{enumerate}[leftmargin=*, itemsep=2pt]
    \item \toolname currently focuses on relational DBMS and \texttt{SELECT} queries like prior works~\cite{jiang2024detecting, ba2023testing, rigger2020detecting, rigger2020testing, rigger2020finding, zhang2025constant, ba2024keep}, but the underlying principles can be extended to other types of databases and query languages, such as graph DBMS~\cite{mang2025finding, liu2024testing, zhuang2023testing,jiangy2024detecting} and spatial DBMS testing~\cite{deng2024finding}.
    \item \toolname focuses on automatically discovering test oracles, but does not provide a mechanism to prioritize or rank the generated oracles. We believe that, after \toolname, DBMS testing research can shift from manually crafting test oracles to developing techniques for prioritizing and selecting the most effective oracles from a large pool of LLM-generated candidates.
\end{enumerate}


\begin{section}{Related Work}
\paragraph{LLM in systems research} 
LLMs have been applied to various system research tasks, including code generation~\cite{hong2025autocomp, mundler2025type, li2025s, xia2024automated, patilberkeley}, automated tuning~\cite{giannakouris2025lambda, zhang2023using}, data processing~\cite{zeighami2025cut, lin2025querying, lin2024towards, patel2024semantic, shankar2024docetl}, program analysis~\cite{li2024enhancing}, networking~\cite{zhang2024caravan}, and algorithm discovery~\cite{novikov2025alphaevolve, rybin2025xx, yu2025autonomous}. 
\begin{enumerate}[leftmargin=*, itemsep=2pt]
    \item In LLM-aided testing, prior work has primarily focused on generating effective inputs. For example, Fuzz4all~\cite{xia2024fuzz4all} is a universal fuzzing framework that leverages LLMs to generate test inputs for various applications, while other research applies LLMs to fuzzing in specific domains such as compilers~\cite{yang2024whitefox}, kernels~\cite{yang2025kernelgpt}, and smart contracts~\cite{shou2024llm4fuzz}. Similarly, ShQvel~\cite{zhong2025testing} uses LLMs to generate feature-rich SQL queries for testing DBMS with manually crafted test oracles, while SQLStorm~\cite{schmidt2025sqlstorm} employs LLMs to synthesize a large-scale benchmark for evaluating performance of DBMSs.
    While prior approaches generate complex test inputs to find crash bugs, our work focuses on generating test oracles to detect logic bugs, along with effective methods for test oracle initializations, which are orthogonal and complementary to existing efforts.
    \item In LLM-for-DBMS, prior work mainly focus on text-to-SQL~\cite{hong2024next}, query optimization~\cite{li2024llm, tan2025can}, dialect translation~\cite{zhou2025cracksql}, and hyperparameter tuning~\cite{giannakouris2025lambda, zhou2024db}. 
    Among these works, LLM-R2~\cite{li2024llm} provides a framework to guide LLMs' use of existing rules, \emph{e.g.,} Calcite, to optimize SQL queries.
    Though this method can also be applied to our task, it is impractical because of the high cost of rewriting concrete queries for testing and the difficulty of finding bugs by DBMS's query rewriting rules, which are typically simple and well-tested.
    To the best of our knowledge, our work is the first to leverage LLMs to generate test oracles for DBMSs, and we believe that our rule generation framework can also be applied to other DBMS tasks in the future, such as query optimization with performance improvement.
\end{enumerate}

\paragraph{DBMS testing and verification} 
There is a rich body of work on DBMS testing, verification, and their applications.
    Prior work in DBMS testing has largely followed two distinct research paths: \emph{test query generation} for fuzzing and \emph{test oracle design} for detecting logic bugs.
    Test query generation focuses on automatically generating complex SQL queries to find crashes, with techniques being either grammar-based, exemplified by tools like SQLsmith~\cite{Seltenreich2022sqlsmith} and SQLancer~\cite{rigger2020testing}, or mutation-based, as seen in Griffin~\cite{fu2023griffin}. Grammar-based techniques typically use predefined grammars to construct valid SQL queries, while mutation-based approaches modify existing queries to create new test cases.
    \modifyc{Beyond the generation of test queries, recent research has also explored constructing meaningful database states~\cite{arasu2011data,binnig2007qagen,chen2025parseval,sanghi2023synthetic,yang2022sam}.}
    In parallel, the second path concentrates on test oracles for detecting ``silent'' bugs such as logic and performance issues, an area that has traditionally relied on significant manual effort. For instance, oracles like TLP~\cite{rigger2020finding}, NoREC~\cite{rigger2020detecting}, and EET~\cite{jiang2024detecting} were designed for logic bugs, while others such as CERT~\cite{ba2024cert} and Apollo~\cite{jung2019apollo} target performance issues. This manual crafting of oracles remains a primary bottleneck in fully automated DBMS testing. \toolname is the first work that tackles this critical challenge by automatically discovering test oracles.

    Recent work has focused on improving the effectiveness of SQL equivalence provers~\cite{zhou2020spes, wang2024qed,ding2023proving,chu2017cosette,chu2017hottsql} and disprovers~\cite{he2024verieql,zhao2025polygon}. These tools have been widely applied in various DBMS scenarios, including query rewriting~\cite{wang2022wetune}, text-to-SQL~\cite{yang2025automated}, and user-database interaction~\cite{purich2025adaptive}. Among these, Wetune~\cite{wang2022wetune} is the most relevant to our work, as it employs an SQL equivalence prover to verify query rewriting rules generated via a brute-force search. However, its enumeration-based approach suffers from two key limitations. First, it can only discover simple rules due to the high time complexity of enumeration. Second, because it lacks a placeholder mechanism like CAQ, the resulting rules support a very limited set of SQL features.
    \modifyc{For example, the average AST node count of SQL queries generated by Wetune’s rules is 99.6, whereas in \toolname test cases it reaches 777.8. This indicates that \toolname targets much richer and complex SQL structures than enumeration-based tools.}
    These restrictions render the approach impractical for testing modern DBMSs. In contrast, our LLM-based approach overcomes these shortcomings by generating complex rules and instantiating them with a diverse range of SQL snippets.
\end{section}
\section{Conclusion}
The manual construction of test oracles is a major bottleneck in finding logic bugs in modern DBMSs.
We present \toolname, the first automated framework that uses LLMs to generate sound and effective test oracles. \toolname achieves scalability through its novel CAQ abstraction and guarantees soundness with an SQL equivalence prover. On five extensively-tested DBMSs, \toolname discovered \bugsnumber previously unknown bugs, including \logicnumber logic bugs. 
Compared to prior work, \toolname's LLM-generated test oracles demonstrate improved metamorphic coverage and a stronger logic-bug finding abilities.
\toolname's approach not only provides a solution to the long-standing bottleneck of fully automated DBMS testing but also opens new opportunities for testing other complex software with LLMs.

\section{Ackowledgements}
We thank Jinsheng Ba and Manuel Rigger for their insightful feedback, Haoran Ding, Sicheng Pan, and Shuxian Wang for their support on SQL equivalence solvers, and the DBMS developers for confirming and fixing the reported bugs. This work is supported by NSF grants IIS-1955488, IIS-2027575, DOE awards DE-SC0016260, AC02-05CH11231, and DARPA Agreement No. HR00112590131.

\bibliographystyle{ACM-Reference-Format}

\clearpage
\bibliography{references}

\clearpage
\appendix
\counterwithin{lstlisting}{section}
\setcounter{lstlisting}{0}

\section{Omitted Algorithm}
\label{sec:appendix-algorithm}
\begin{algorithm}[h!]
    \caption{\small $\mathsf{GuidedCAQPairsSampling}$}
    \small
    \KwIn{
        A DBMS $\mathcal{D}$, a set of equivalent CAQs $\mathsf{Equal}$, a number of maximal number of samples to return $\mathsf{MaxSamples}$, a number of maximal iterations for clustering $\mathsf{MaxIters}$.
    }
    \KwOut{A sampled set of equivalent CAQ pairs $\mathsf{Sam}$.}
    \SetKwFunction{FMain}{$\mathsf{Distance}$}
    \SetKwProg{Fn}{Function}{:}{}
    \Fn{\FMain{$\mathcal{D}, \mathsf{caq}, \mathsf{cluster}$}}{
        $\mathsf{MinDistance}$ $\leftarrow +\infty$\;
        \ForEach{$\mathsf{q} \in \mathsf{cluster}$}{
            $\mathsf{plan}_1$ $\leftarrow \mathcal{D}.\mathsf{Explain}(\mathsf{caq})$\;
            $\mathsf{plan}_2$ $\leftarrow \mathcal{D}.\mathsf{Explain}(\mathsf{q})$\;
            $\mathsf{dist}$ $\leftarrow \mathsf{TreeEditDistance}(\mathsf{plan}_1, \mathsf{plan}_2)$\;
            \If{$\mathsf{dist} < \mathsf{MinDistance}$}{
                $\mathsf{MinDistance} \leftarrow \mathsf{dist}$\;
            }
        }
        \Return $\mathsf{MinDistance}$\;
    }
    \If{$|\mathsf{Equal}| \leq \mathsf{MaxSamples}$}{
        $\mathsf{Sam} \leftarrow \mathsf{Equal}$\;
        \Return $\mathsf{Sam}$\;
    }
    $\mathsf{Clusters} \leftarrow \mathsf{RandomSample}(\mathsf{Equal}, \mathsf{MaxSamples})$;\\
    \For{$\mathsf{iter} \leftarrow 1$ \KwTo $\mathsf{MaxIters}$}{
        $\mathsf{NewClusters} \leftarrow \{\}$\;
        \ForEach{$\mathsf{caq} \in \mathsf{Equal}$}{
            $i^* \leftarrow \arg\min_i \mathsf{Distance}(\mathcal{D}, \mathsf{caq}, \mathsf{Clusters}[i])$\;
            $\mathsf{NewClusters}[i^*] \leftarrow \mathsf{NewClusters}[i^*] \cup \{\mathsf{caq}\}$\;
        }
        \If{$\mathsf{NewClusters} = \mathsf{Clusters}$}{
            break\;
        }
        $\mathsf{Clusters} \leftarrow \mathsf{NewClusters}$\;
    }
    $\mathsf{Sam} \leftarrow \{\}$\;
    \ForEach{$\mathsf{cluster} \in \mathsf{Clusters}$}{
        $\mathsf{caq} \leftarrow \mathsf{RandomSample}(\mathsf{cluster}, 1)$\;
        $\mathsf{Sam} \leftarrow \mathsf{Sam} \cup \{\mathsf{caq}\}$\;
    }
    \Return $\mathsf{Sam}$\;
\end{algorithm}
\section{Counter-Examples and Proofs}
\label{sec:appendix-proof}

In this section, we provide three counter-examples to demonstrate that the potential false positives if we do not apply the three general constraints of $\mathcal{C}$ mentioned in \Cref{sec:query-instantiation}.
Specifically, Listing~\ref{lst:counter-example-1}, Listing~\ref{lst:counter-example-2}, and Listing~\ref{lst:counter-example-3} show three cases that can be formally proven in the scenario of virtual columns, but lead to false positives in the instantiation phase by missing the constraints.

\begin{figure}[h!]
\centering
\begin{lstlisting}[
    linewidth=\linewidth, % Set listing width to match the minipage
    escapeinside={(*}{*)},
    caption={False positive example when instantiating non-deterministic snippets.},
    captionpos=b,
    label={lst:counter-example-1}
]
CREATE TABLE t1(c0 BOOLEAN);
INSERT INTO t1 VALUES (TRUE);
SELECT * FROM t1 
WHERE (*${\square_1}\triangleright\placeholder{\textsf{Expr}(t:\textsf{BOOLEAN})} \mapsto$*) (*\texttt{\textcolor{NavyBlue}{RANDOM() < 0.5}}*);
-- {unstably return 1 row or 0 rows}
SELECT * FROM t1 
WHERE (*${\square_1}\triangleright\placeholder{\textsf{Expr}(t:\textsf{BOOLEAN})} \mapsto$*) (*\texttt{\textcolor{NavyBlue}{RANDOM() < 0.5}}*);
-- {unstably return 1 row or 0 rows}
\end{lstlisting}
\end{figure}

\begin{figure}[h!]
\centering
\begin{lstlisting}[
    linewidth=\linewidth, % Set listing width to match the minipage
    escapeinside={(*}{*)},
    caption={False positive example when instantiating snippets without null-preserving property.},
    captionpos=b,
    label={lst:counter-example-2}
]
CREATE TABLE t1(c0 INT);
CREATE TABLE t2(c0 INT);
INSERT INTO t1 VALUES (1);
INSERT INTO t2 VALUES (1);
SELECT (*${\square_1}\triangleright\placeholder{\textsf{Expr}(t:\textsf{INT})} \mapsto$*) (*\texttt{\textcolor{NavyBlue}{NULLIF(t2.c0, 0)}}*) FROM t1 LEFT JOIN t2 WHERE false; -- {0}
SELECT NULL FROM t1 LEFT JOIN t2  WHERE false; -- {NULL}
\end{lstlisting}
\end{figure}

\begin{figure}[h!]
\centering
\begin{lstlisting}[
    linewidth=\linewidth, % Set listing width to match the minipage
    escapeinside={(*}{*)},
    caption={False positive example when instantiating snippets without empty-preserving property.},
    captionpos=b,
    label={lst:counter-example-3}
]
CREATE TABLE t1(c0 INT);
INSERT INTO t1 VALUES (1);
SELECT * FROM t1 WHERE false; -- {0 rows}
SELECT * FROM t1 WHERE EXISTS (
    SELECT (*${\square_1}\triangleright\placeholder{\textsf{Expr}(t:\textsf{INT})} \mapsto$*) (*\texttt{\textcolor{NavyBlue}{SUM(c0)}}*) FROM t1 WHERE false ) ; -- {1 row}
\end{lstlisting}
\end{figure}

We use these constraints to ensure that snippets replacing column expressions are deterministic, preserve null values during outer joins, and refer to a unique row's column combination in the cases other than outer joins.
We then prove the following theorem to show that if we follow these constraints, the equivalence is preserved after instantiating expressions.
\begin{theorem}[Expression Instantiation Preserves Equivalence]
Consider two SQL queries, $Q_1$ and $Q_2$, that are \textbf{semantically equivalent} on a schema $S$ containing virtual columns $\{v_i\}$. Let $Q_1'$ and $Q_2'$ be the queries created by replacing every occurrence of each virtual column $v_i$ with its corresponding expression $e_i$. If each expression $e_i$ references a \textbf{unique row's} column combination in all cases other than outer joins, then $Q_1'$ and $Q_2'$ are also \textbf{semantically equivalent} on the schema $S'$, which is $S$ without the virtual columns.
\end{theorem}
\begin{proof}
We prove this by contradiction. Assume that $Q_1'$ and $Q_2'$ are not semantically equivalent on schema $S'$, meaning there exists a database instance $D'$ on $S'$ where $Q_1'(D') \neq Q_2'(D')$. We then construct a database $D$ on schema $S$ by augmenting $D'$ with the virtual columns $\{v_i\}$ defined by their corresponding expressions $\{e_i\}$. The core of this proof is establishing that for any query $Q_k$, the results $Q_k(D)$ and $Q_k'(D')$ are identical. This equivalence is guaranteed by two properties of the expressions $\{e_i\}$. First, the \textbf{unique row reference} property ensures that for any standard join or matched row, the value of a virtual column $v_i$ in $D$ is computed from the same single row as its expression $e_i$ in $D'$, yielding an identical result. Second, and crucially for outer joins, the \textbf{null-preserving} property handles unmatched rows. If an outer join produces a row where the columns underlying an expression $e_i$ are padded with \texttt{NULL}, the expression $e_i$ itself evaluates to \texttt{NULL} in query $Q_k'$. The null-preserving property explicitly forces the corresponding virtual column $v_i$ in query $Q_k$ to also resolve to \texttt{NULL} in this exact scenario. 
This perfect mirroring of behavior ensures that $Q_1(D) = Q_1'(D')$ and $Q_2(D) = Q_2'(D')$. 
However, since $Q_1'(D') \neq Q_2'(D')$ by our initial assumption, it follows that $Q_1(D) \neq Q_2(D)$, contradicting the premise that $Q_1$ and $Q_2$ are semantically equivalent on schema $S$.
\end{proof}

This theorem shows that, under the scenario of only virtual columns, the equivalence is preserved after instantiating expressions if we follow the constraints mentioned in \Cref{sec:query-instantiation}.
When there are virtual tables, it is much easier to show the preserved equivalence after instantiating tables, since the tables are simply replaced by other tables with the same schema.
If there is no non-deterministic rows in the instantiated tables, the equivalence is also preserved after instantiating tables, because such instantiation is equivalent to using a \texttt{CREATE VIEW} to replace the \texttt{CREATE TABLE} for virtual tables.

Note that, while these proofs do not cover all SQL features, especially for those new features in future SQL standards, we have not found any counterexamples in our experiments that violate the equivalence after instantiating expressions following these constraints.
In practice, we can maintain the soundness of instantiated queries by disabling any new features that invalidate existing proofs.
\section{Representing Existing Oracles}

\label{sec:appendix-oracle}
\begin{figure}[h!]
\begin{lstlisting}[
    linewidth=\linewidth, % Set listing width to match the minipage
    escapeinside={(*}{*)},
    caption={TLP~\cite{rigger2020finding} in CAQ.},
    captionpos=b,
    label={lst:tlp-in-caq-appendix}
]
(1) SELECT * FROM t1, (*${\square_1}\triangleright\placeholder{\textsf{Table}(...)}$*);
(2) SELECT * FROM t1, (*${\square_1}\triangleright\placeholder{\textsf{Table}(...)}$*) WHERE (*${\square_2}\triangleright\placeholder{\textsf{Expr}(t1:\textsf{BOOLEAN})}$*) AND (*${\square_1}\triangleright\placeholder{\textsf{Table}(...)}$*).c0 IS TRUE UNION ALL 
    SELECT * FROM t1, (*${\square_1}\triangleright\placeholder{\textsf{Table}(...)}$*) WHERE (*${\square_2}\triangleright\placeholder{\textsf{Expr}(t1:\textsf{BOOLEAN})}$*) AND (*${\square_1}\triangleright\placeholder{\textsf{Table}(...)}$*).c0 IS FALSE UNION ALL 
    SELECT * FROM t1, (*${\square_1}\triangleright\placeholder{\textsf{Table}(...)}$*) WHERE (*${\square_2}\triangleright\placeholder{\textsf{Expr}(t1:\textsf{BOOLEAN})}$*) AND (*${\square_1}\triangleright\placeholder{\textsf{Table}(...)}$*).c0 IS NULL;
\end{lstlisting}
\end{figure}
\begin{figure}[h!]
\begin{lstlisting}[
    linewidth=\linewidth, % Set listing width to match the minipage
    escapeinside={(*}{*)},
    caption={NoREC~\cite{rigger2020detecting} in CAQ.},
    captionpos=b,
    label={lst:norec-in-caq-appendix}
]
(1) SELECT COUNT((***)) FROM t1, (*${\square_1}\triangleright\placeholder{\textsf{Table}(...)}$*) WHERE (*${\square_2}\triangleright\placeholder{\textsf{Expr}(t1:\textsf{BOOLEAN})}$*) AND (*${\square_1}\triangleright\placeholder{\textsf{Table}(...)}$*).c0; 
(2) SELECT SUM(cnt) FROM (SELECT ((*${\square_2}\triangleright\placeholder{\textsf{Expr}(t1:\textsf{BOOLEAN})}$*) AND (*${\square_1}\triangleright\placeholder{\textsf{Table}(...)}$*).c0) 
    IS TRUE AS cnt FROM t1, (*${\square_1}\triangleright\placeholder{\textsf{Table}(...)}$*));
\end{lstlisting}
\end{figure}
\begin{figure}[h!]
\begin{lstlisting}[
    linewidth=\linewidth, % Set listing width to match the minipage
    escapeinside={(*}{*)},
    caption={EET~\cite{jiang2024detecting} in CAQ.},
    captionpos=b,
    label={lst:eet-in-caq-appendix}
]
(1) SELECT * FROM t1, (*${\square_1}\triangleright\placeholder{\textsf{Table}(...)}$*) WHERE (*${\square_2}\triangleright\placeholder{\textsf{Expr}(t1:\textsf{BOOLEAN})}$*);
(2) SELECT * FROM t1, (*${\square_1}\triangleright\placeholder{\textsf{Table}(...)}$*) WHERE (((*${\square_1}\triangleright\placeholder{\textsf{Table}(...)}$*).c0 AND (NOT (*${\square_1}\triangleright\placeholder{\textsf{Table}(...)}$*).c0) 
    AND ((*${\square_1}\triangleright\placeholder{\textsf{Table}(...)}$*).c0 IS NOT NULL))) OR (*${\square_2}\triangleright\placeholder{\textsf{Expr}(t1:\textsf{BOOLEAN})}$*);
(3) SELECT * FROM t1, (*${\square_1}\triangleright\placeholder{\textsf{Table}(...)}$*) WHERE CASE WHEN 
    (((*${\square_1}\triangleright\placeholder{\textsf{Table}(...)}$*).c0 AND (NOT (*${\square_1}\triangleright\placeholder{\textsf{Table}(...)}$*).c0) AND ((*${\square_1}\triangleright\placeholder{\textsf{Table}(...)}$*).c0 IS NOT NULL)))
    THEN (*${\square_3}\triangleright\placeholder{\textsf{Expr}(t1:\textsf{BOOLEAN})}$*) ELSE (*${\square_2}\triangleright\placeholder{\textsf{Expr}(t1:\textsf{BOOLEAN})}$*) END;
(4) SELECT * FROM t1, (*${\square_1}\triangleright\placeholder{\textsf{Table}(...)}$*) WHERE CASE WHEN 
    (*${\square_1}\triangleright\placeholder{\textsf{Table}(...)}$*).c0 THEN (*${\square_2}\triangleright\placeholder{\textsf{Expr}(t1:\textsf{BOOLEAN})}$*) ELSE (*${\square_2}\triangleright\placeholder{\textsf{Expr}(t1:\textsf{BOOLEAN})}$*) END;
\end{lstlisting}
\end{figure}
\begin{figure}[h!]
\centering
\begin{lstlisting}[
    linewidth=\linewidth, % Set listing width to match the minipage
    escapeinside={(*}{*)},
    caption={DQP~\cite{ba2024keep} in CAQ.},
    captionpos=b,
    label={lst:dqp-in-caq-appendix}
]
(1) SELECT * FROM t1 JOIN (*${\square_1}\triangleright\placeholder{\textsf{Table}(...)}$*) ON (*${\square_2}\triangleright\placeholder{\textsf{Expr}(t1:\textsf{ANY})}$*) = (*${\square_1}\triangleright\placeholder{\textsf{Table}(...)}$*).c0;  
(2) /*+ JOIN_ORDER((*${\square_1}\triangleright\placeholder{\textsf{Table}(...)}$*), t1) */ 
    SELECT * FROM t1 JOIN (*${\square_1}\triangleright\placeholder{\textsf{Table}(...)}$*) ON (*${\square_2}\triangleright\placeholder{\textsf{Expr}(t1:\textsf{ANY})}$*) = (*${\square_1}\triangleright\placeholder{\textsf{Table}(...)}$*).c0; 
(3) /*+ LOOKUP_JOIN((*${\square_1}\triangleright\placeholder{\textsf{Table}(...)}$*), t1) */
    SELECT * FROM t1 JOIN (*${\square_1}\triangleright\placeholder{\textsf{Table}(...)}$*) ON (*${\square_2}\triangleright\placeholder{\textsf{Expr}(t1:\textsf{ANY})}$*) = (*${\square_1}\triangleright\placeholder{\textsf{Table}(...)}$*).c0; 
(4) /*+ MERGE_JOIN((*${\square_1}\triangleright\placeholder{\textsf{Table}(...)}$*), t1) */
    SELECT * FROM t1 JOIN (*${\square_1}\triangleright\placeholder{\textsf{Table}(...)}$*) ON (*${\square_2}\triangleright\placeholder{\textsf{Expr}(t1:\textsf{ANY})}$*) = (*${\square_1}\triangleright\placeholder{\textsf{Table}(...)}$*).c0;
(5) /*+ HASH_JOIN((*${\square_1}\triangleright\placeholder{\textsf{Table}(...)}$*), t1) */
    SELECT * FROM t1 JOIN (*${\square_1}\triangleright\placeholder{\textsf{Table}(...)}$*) ON (*${\square_2}\triangleright\placeholder{\textsf{Expr}(t1:\textsf{ANY})}$*) = (*${\square_1}\triangleright\placeholder{\textsf{Table}(...)}$*).c0; 
(6) /*+ INNER_JOIN((*${\square_1}\triangleright\placeholder{\textsf{Table}(...)}$*), t1) */
    SELECT * FROM t1 JOIN (*${\square_1}\triangleright\placeholder{\textsf{Table}(...)}$*) ON (*${\square_2}\triangleright\placeholder{\textsf{Expr}(t1:\textsf{ANY})}$*) = (*${\square_1}\triangleright\placeholder{\textsf{Table}(...)}$*).c0;
\end{lstlisting}
\end{figure}

\newtcblisting{sqlbox}[1][]{%
    colback=gray!5,
    colframe=gray!75!black,
    title=#1,
    fonttitle=\bfseries\footnotesize,
    listing only,
    listing engine=listings,
    left=6mm,
    breakable,
    listing options={
        language=SQL,
        basicstyle=\ttfamily\footnotesize,
        numbers=left,
        numberstyle=\tiny,
        breaklines=true,
        columns=fullflexible,
        keepspaces=true,
        showstringspaces=false,
        escapeinside=||
    },
    #1
}

\section{Generated SQL Queries}
\label{sec:generated-sql-queries}

\begin{sqlbox}[title = Query 1 of Argus]
SELECT t8.c3,
       t0.c2,
       t1.c3,
       t1.c2,
       t8.c0,
       t0.c3,
       (CASE WHEN ( (div(t3.c3 + 10, nullif(( EXTRACT(DAY FROM (DATE '2000-01-01' + t3.c3 * INTERVAL '1 day'))::int ), 1)) * mod(( EXTRACT(DAY FROM (DATE '2000-01-01' + t3.c3 * INTERVAL '1 day'))::int ) + 5, 3))::int ) = 0 THEN t3.c0 ELSE t3.c2 END),
       t8.c1,
       ((~ (~ ((-1518331885)&(( (+ ((t8.c2)/(((2087249086)#(760045829))))) )))))),
       vtable0.c0,
       t3.c3,
       vtable0.c3,
       (CASE WHEN t1.c1 IS NULL THEN ( char_length(( (REPLACE(UPPER(t1.c1), 'A', '4') || LPAD(t1.c1::TEXT, 3, '0'))::VARCHAR )) * t1.c0 ) > 0 ELSE length(t1.c1) > 3 END)
FROM t0
CROSS JOIN t8
CROSS JOIN t1
CROSS JOIN t3
CROSS JOIN vtable0
WHERE 1<>1;
\end{sqlbox}

\begin{sqlbox}[title = Query 2 of Argus]
SELECT dt.t0_c2,
       vtable0.c2,
       dt.t0_c1,
       dt.t8_c1
FROM t1
LEFT JOIN vtable0
  ON NOT (t1.c1 BETWEEN (md5(t1.c1 || ( regexp_replace(t1.c1, '[aeiou]', '*', 'gi') ))) AND t1.c1)
CROSS JOIN (
  SELECT t0.c2 AS t0_c2,
         t0.c1 AS t0_c1,
         t8.c1 AS t8_c1
  FROM t8
  INNER JOIN t0
    ON t8.c1 <> (t8.c1 || t0.c1)
) AS dt
WHERE (( COALESCE(length(( to_char(t1.c2 * t1.c3, 'FM9999999')::varchar )), 0) - ( EXTRACT(MONTH FROM (date '2000-01-01' + t1.c3 * interval '1 day'))::int ) ) = ANY(ARRAY[t1.c3, t1.c3 * 2, 0])) IS NULL;
\end{sqlbox}

\begin{sqlbox}[title = Query 3 of Argus]
SELECT (bit_count(( translate(md5(( (CASE WHEN t0.c0 THEN t0.c1 ELSE t0.c0 END) )), 'abcdef', 'FEDCBA') )::bytea)::integer),
       (translate(( (t1.c3::text || ':' || ( ( EXTRACT(EPOCH FROM (TO_TIMESTAMP('2000-01-01','YYYY-MM-DD') + (t1.c3 || ' days')::INTERVAL))::INT ) + t1.c3 * 3 - (t1.c3 & t1.c3) )::text) ), 'aeiou', '12345')::varchar),
       t1.c3,
       (t3.c3 % 2 = 0),
       (REPLACE(t0.c1, t0.c0, REVERSE(t0.c0 || t0.c1))),
       vtable0.c2,
       vtable0.c1,
       t1.c0,
       vtable0.c3,
       vtable1.c2,
       (bit_length(t0.c1) / 2 + octet_length(( LPAD(t0.c1, 10, SUBSTRING(t0.c0 FROM 1 FOR 1)) ))),
       t1.c2,
       t3.c0
FROM t0, t1, vtable0, vtable1, t3
WHERE FALSE;
\end{sqlbox}

\begin{sqlbox}[title = Query 4 of Argus]
SELECT
  vtable0.c3,
  t3.c3,
  (get_bit(decode(to_hex(t3.c3), 'hex'), 7) = 1),
  vtable0.c0,
  (char_length(t8.c1) + char_length(coalesce(t8.c3,''))),
  t1.c3,
  vtable0.c2,
  t8.c0,
  (LENGTH(md5(COALESCE(( translate(t0.c1 || t0.c0, 'aeiouAEIOU', '1234512345') ), t0.c0)))),
  (to_char(DATE '2020-01-01' + (t3.c3 || ' days')::interval, 'YYYY-MM-DD')::VARCHAR),
  (regexp_count(( t0.c0 || '_' || ( t0.c0 || '-' || t0.c1 ) ), '[0-9]')),
  t8.c1,
  t0.c2,
  vtable1.c1,
  vtable1.c2,
  t8.c2,
  t3.c0,
  vtable0.c1,
  t0.c0,
  t0.c1,
  ((( encode(sha256(convert_to(( trim(both '#' from t8.c3)::varchar ), 'UTF8')), 'hex') ) LIKE '%test%' OR t8.c3 ~ '^[0-9]+$') AND t8.c2 BETWEEN 1 AND 10),
  ((t3.c2 AND NOT t3.c1) OR (( extract(epoch from make_interval(days => t3.c3))::int ) % 2 = 0)),
  ((CASE WHEN ( (SELECT bool_and(x::boolean) FROM jsonb_array_elements_text(jsonb_build_array(t0.c0, t0.c2)) AS a(x)) ) THEN t0.c0 ELSE ( translate(t0.c0, 'aeiou', '12345') ) END)),
  (regexp_count(( t8.c1 || '_' || ( t8.c1 || '-' || t8.c3 ) ), '[0-9]')),
  t1.c2,
  t0.c3,
  (reverse(t1.c1)),
  t8.c3,
  vtable1.c0,
  t1.c1,
  (CASE WHEN t1.c3 > t1.c2 THEN bit_length(t1.c2) ELSE width_bucket(t1.c2, 0, 100, 10) END),
  t3.c2,
  (((((((( regexp_replace(t1.c3 || ( to_char(( (+ (+ ((( (t1.c0 << 2) + (t1.c3 >> 1) ))*(302076923)))) ), 'FM0000') ), '[aeiou]', '*', 'gi') ))||('')))||(((t1.c3)||('rR')))))<=((((('')||('h')))||(((t1.c3)||('5M'))))))),
  vtable1.c3,
  t3.c1
FROM
  vtable0
  LEFT JOIN t8 ON vtable0.c0
  CROSS JOIN t3
  CROSS JOIN vtable1
  CROSS JOIN t1
  CROSS JOIN t0
WHERE
  ((t0.c2 OR t0.c3) AND t0.c3) = (t0.c3 < t0.c2)
;
\end{sqlbox}

\begin{sqlbox}[title = Query 5 of Argus]
SELECT
  t0.c3,
  sub_t1.c0_val,
  vtable1.c2,
  t0.c0,
  (CHAR_LENGTH(t0.c0 || ( REPEAT(t0.c1, GREATEST(1, LEAST(( bit_length(t0.c0) ), 5))) ))),
  (octet_length(translate(t0.c0, 'aeiou', 'AEIOU'))),
  sub_t1.ph1,
  vtable1.c0
FROM t0
CROSS JOIN vtable1
LEFT JOIN (
  SELECT
    t1.c0 AS c0_val,
    (make_interval(days => ( width_bucket(t1.c3, 1, 100, 5) )) < make_interval(days => t1.c2)) AS ph1
  FROM t1
) sub_t1 ON vtable1.c2;
\end{sqlbox}

\begin{sqlbox}[title = Query 6 of Argus]
SELECT t8.c3,
       t0.c2,
       t1.c3,
       t1.c2,
       t8.c0,
       t0.c3,
       ((t3.c3 % 2 = 0) OR t3.c1),
       t8.c1,
       (COALESCE(length(t8.c2), 0) - ( EXTRACT(MONTH FROM (date '2000-01-01' + ( GREATEST(( EXTRACT(EPOCH FROM (DATE '2020-01-01' + (t8.c2 || ' days')::interval))::int ), -( EXTRACT(DAY FROM ('2000-01-01'::timestamp + t8.c2 * '1 day'::interval))::int )) ) * interval '1 day'))::int )),
       vtable0.c0,
       t3.c3,
       vtable0.c3,
       (t1.c0 > ( (t1.c0 + COALESCE(( t1.c0 + ( EXTRACT(EPOCH FROM (timestamp '2000-01-01' + (t1.c0 || ' days')::interval))::int ) * LENGTH(t1.c2) ), 0)) * GREATEST(t1.c0 - t1.c2, 1) ))
FROM t0
CROSS JOIN t8
CROSS JOIN t1
CROSS JOIN t3
CROSS JOIN vtable0
LIMIT 0;
\end{sqlbox}

\begin{sqlbox}[title = Query 7 of Argus]
SELECT t8.c2,
       t8.c3,
       t8.c0,
       t0.c2,
       t0.c1,
       vtable0.c1,
       vtable0.c0,
       ((CASE WHEN ( t0.c0 LIKE '%test%' ) THEN t0.c1 ELSE t0.c0 END)),
       t0.c3,
       vtable0.c3,
       (ascii(SUBSTRING(( ( (t0.c1 || '_' || ( reverse(substring(( regexp_replace(t0.c1, '[0-9]', 'X', 'g') ) from 1 for 3))::varchar ))::varchar ) || '_' || t0.c0 ) FROM 1 FOR 1)) * CASE WHEN ( regexp_like(t0.c0, 'test$', 'i') ) THEN 1 ELSE 2 END),
       (bit_length(( replace(lower(t8.c2),'a','@') || to_char(t8.c2,'FM9999') ))),
       (( split_part(t8.c0 || ':' || ( md5(t8.c0 || t8.c3)::varchar ), ':', 2) ) ~ '^[0-9]+$' AND length(( md5(t8.c0 || t8.c3)::varchar )) > 0)
FROM t8
CROSS JOIN t0
JOIN vtable0
  ON ((ascii(substring(t0.c0 FROM 1 FOR 1)) * (length(coalesce(( translate(t0.c0 || t0.c1, 'aeiou', 'AEIOU') ), '')) + 1)) + vtable0.c1) = ((ascii(substring(t0.c0 FROM 1 FOR 1)) * (length(coalesce(( translate(t0.c0 || t0.c1, 'aeiou', 'AEIOU') ), '')) + 1)) + -492369572)
WHERE vtable0.c3 IS NULL
   OR (ascii(SUBSTRING(( ( (t0.c1 || '_' || ( reverse(substring(( regexp_replace(t0.c1, '[0-9]', 'X', 'g') ) from 1 for 3))::varchar ))::varchar ) || '_' || t0.c0 ) FROM 1 FOR 1)) * CASE WHEN ( regexp_like(t0.c0, 'test$', 'i') ) THEN 1 ELSE 2 END) IS NULL
UNION ALL
SELECT t8.c2,
       t8.c3,
       t8.c0,
       t0.c2,
       t0.c1,
       NULL,
       NULL,
       ((CASE WHEN ( t0.c0 LIKE '%test%' ) THEN t0.c1 ELSE t0.c0 END)),
       t0.c3,
       NULL,
       (ascii(SUBSTRING(( ( (t0.c1 || '_' || ( reverse(substring(( regexp_replace(t0.c1, '[0-9]', 'X', 'g') ) from 1 for 3))::varchar ))::varchar ) || '_' || t0.c0 ) FROM 1 FOR 1)) * CASE WHEN ( regexp_like(t0.c0, 'test$', 'i') ) THEN 1 ELSE 2 END),
       (bit_length(( replace(lower(t8.c2),'a','@') || to_char(t8.c2,'FM9999') ))),
       (( split_part(t8.c0 || ':' || ( md5(t8.c0 || t8.c3)::varchar ), ':', 2) ) ~ '^[0-9]+$' AND length(( md5(t8.c0 || t8.c3)::varchar )) > 0)
FROM t8
CROSS JOIN t0
WHERE NOT EXISTS (
  SELECT 1
  FROM vtable0
  WHERE ((ascii(substring(t0.c0 FROM 1 FOR 1)) * (length(coalesce(( translate(t0.c0 || t0.c1, 'aeiou', 'AEIOU') ), '')) + 1)) + vtable0.c1) = ((ascii(substring(t0.c0 FROM 1 FOR 1)) * (length(coalesce(( translate(t0.c0 || t0.c1, 'aeiou', 'AEIOU') ), '')) + 1)) + -492369572)
);
\end{sqlbox}

\begin{sqlbox}[title = Query 8 of Argus]
SELECT (char_length(t0.c1 || t0.c0)),
       (substring(t1.c0 from '[A-Za-z]+') || '-' || to_hex(( (t1.c2 / GREATEST(1, t1.c0)) % 7 ))),
       t1.c3,
       (t3.c3 % 2 = 0),
       ((t0.c1 || '-' || ( translate(t0.c1, 'aeiou', '12345') ))),
       vtable0.c2,
       vtable0.c1,
       t1.c0,
       vtable0.c3,
       vtable1.c2,
       (char_length(t0.c0 || t0.c1)),
       t1.c2,
       t3.c0
FROM t0, t1, vtable0, vtable1, t3
WHERE FALSE;
\end{sqlbox}

\begin{sqlbox}[title = Query 9 of Argus]
SELECT
  vtable0.c3,
  t3.c3,
  ((t3.c3::text) SIMILAR TO '[0-9]+'),
  vtable0.c0,
  (mod(coalesce(( (bit_length(t8.c1::bytea) + ascii(substring(( REGEXP_REPLACE(t8.c1, 'a', 'b', 'g') ) from 1 for 1))) ),0) * coalesce(( regexp_count(( to_char(DATE '2021-01-01' + t8.c1 * INTERVAL '1 day', 'YYYY-MM-DD') ), '\\w+', 1, 'g') ),1), 10) + regexp_instr(t8.c1, '[0-9]')),
  t1.c3,
  vtable0.c2,
  t8.c0,
  (CASE WHEN ( t0.c1 @@ plainto_tsquery(t0.c0) ) THEN ( get_byte(t0.c0::bytea, 1) + get_byte(t0.c1::bytea, 0) ) ELSE ( regexp_count(t0.c1, '\\w+', 1, 'g') ) END),
  (translate(( to_char(t3.c3, 'FM0000')::varchar ),'abc','ABC')::varchar),
  ((CASE ((((t0.c1)||(t0.c0)))||((CASE t0.c0 WHEN t0.c1 THEN t0.c1 END ))) WHEN (((('')||('')))||(((NULL)||('')))) THEN ((((1574883451)%(NULL)))*((CASE NULL WHEN '+~' THEN 315562183 END ))) END )),
  t8.c1,
  t0.c2,
  vtable1.c1,
  vtable1.c2,
  t8.c2,
  t3.c0,
  vtable0.c1,
  t0.c0,
  t0.c1,
  (( xml_is_well_formed(t8.c3) ) OR NOT ( (( floor((( coalesce(substring(t8.c3 from E'(\\d+)'), '0')::int + t8.c2 )::numeric / greatest(1, t8.c2 + 1)))::int ) % 2 = 0) )),
  (t3.c3 AND ( (t3.c3 % 2 = 0) ) OR NOT t3.c3),
  (t0.c0 || '-' || t0.c1),
  (family('127.0.0.1'::inet) + masklen('127.0.0.1/24'::inet) - t8.c2),
  t1.c2,
  t0.c3,
  (UPPER(SUBSTRING(t1.c1 FROM 1 FOR GREATEST(3, LENGTH(t1.c1) % 5)))),
  t8.c3,
  vtable1.c0,
  t1.c1,
  (bit_count((( EXTRACT(YEAR FROM DATE '1990-01-01' + ( length(t1.c3) + t1.c0 * 2 - floor(t1.c3::double precision / nullif(t1.c0,1))::int ) * INTERVAL '1 day')::INT ) + t1.c3)::bit(16))::int),
  t3.c2,
  ((t1.c2::text) SIMILAR TO '[0-9]+'),
  vtable1.c3,
  t3.c1
FROM
  vtable0
  LEFT JOIN t8 ON vtable0.c0
  CROSS JOIN t3
  CROSS JOIN vtable1
  CROSS JOIN t1
  CROSS JOIN t0
WHERE
  ((t0.c2 OR t0.c3) AND t0.c3) = (t0.c3 < t0.c2)
;
\end{sqlbox}

\begin{sqlbox}[title = Query 10 of Argus]
SELECT v.c0, v.c2
FROM vtable0 v
CROSS JOIN LATERAL (SELECT 1 FROM t1) l1
CROSS JOIN LATERAL (SELECT 1 FROM t8) l2
CROSS JOIN LATERAL (SELECT 1 FROM t0 WHERE c2 IS NULL) l3;
\end{sqlbox}

\begin{sqlbox}[title = Query 1 of SQLancer]
SELECT ALL * 
FROM t1 INNER 
JOIN (SELECT - (+ (t0.c0)) 
FROM ONLY t0, ONLY t1 
WHERE TRUE) AS sub0 ON TRUE;
\end{sqlbox}

\begin{sqlbox}[title = Query 2 of SQLancer]
SELECT t0.c0, t1.c0 
FROM t0* 
FULL OUTER JOIN ONLY t1 ON lower((0.961022777314273)::VARCHAR(181))~*lower(CAST(TRUE AS VARCHAR));
\end{sqlbox}

\begin{sqlbox}[title = Query 3 of SQLancer]
SELECT * 
FROM t1* 
WHERE TRUE 
UNION ALL SELECT ALL * 
FROM t1* 
WHERE NOT (TRUE) 
UNION ALL SELECT ALL * 
FROM t1 
WHERE (TRUE) IS NULL;
\end{sqlbox}

\begin{sqlbox}[title = Query 4 of SQLancer]
SELECT t1.c0 
FROM t1;
\end{sqlbox}

\begin{sqlbox}[title = Query 5 of SQLancer]
SELECT t0.c0, t1.c0 
FROM t1* 
LEFT OUTER JOIN t0 ON inet_same_family('179.19.19.249', '240.210.141.13');
\end{sqlbox}

\begin{sqlbox}[title = Query 6 of SQLancer]
SELECT ALL t1.c0 
FROM t1* 
CROSS JOIN (SELECT ALL (t1.c0) BETWEEN (t0.c0) AND (num_nonnulls('')), (family('118.88.171.159')) IN (t1.c0, t0.c0, t1.c0), CAST(0.07644869 AS MONEY) 
FROM ONLY t0, t1) AS sub0;
\end{sqlbox}

\begin{sqlbox}[title = Query 7 of SQLancer]
SELECT ALL t1.c0 
FROM ONLY t1 
WHERE ((0.021889338)::MONEY) IS NULL 
UNION ALL SELECT t1.c0 
FROM ONLY t1 
WHERE NOT (((0.021889338)::MONEY) IS NULL) 
UNION ALL SELECT t1.c0 
FROM ONLY t1 
WHERE (((0.021889338)::MONEY) IS NULL) ISNULL;
\end{sqlbox}

\begin{sqlbox}[title = Query 8 of SQLancer]
SELECT ALL t0.c0 
FROM t1 
LEFT OUTER JOIN ONLY t0 ON NOT ((((((t0.c0)%(t1.c0)))::VARCHAR)LIKE(((('{zq.O]w\n')::VARCHAR(701))||(rtrim('0.7364646181551185'))))));
\end{sqlbox}

\begin{sqlbox}[title = Query 9 of SQLancer]
SELECT ALL t1.c0 
FROM t1 INNER 
JOIN t0 ON (0.506531) IN (t1.c0, 0.46897623, t1.c0);
\end{sqlbox}

\begin{sqlbox}[title = Query 10 of SQLancer]
SELECT ALL t0.c0 
FROM ONLY t0;
\end{sqlbox}


\end{document}